\documentclass[11pt, a4paper]{article}

\usepackage[utf8]{inputenc}
\usepackage[T1]{fontenc}
\usepackage{lmodern}

\usepackage[a4paper, margin=1in]{geometry}
\usepackage[mathcal]{euscript}
\usepackage{amsmath}    
\usepackage{amssymb}    
\usepackage{amsthm}     
\usepackage{bbm}        
\usepackage{mathtools}  
\usepackage{physics}    
\usepackage{braket}     
\usepackage{dsfont}     
\usepackage{tikz}
\usepackage{tcolorbox}
\usepackage{relsize}	
\usepackage{xparse}  	
\usepackage{MnSymbol}	
\usepackage{authblk}    
\usepackage{graphicx}   
\usepackage{float}      
\usepackage{subcaption} 
\usepackage{xpatch}
\usepackage{setspace}
\usepackage{enumitem}
\usepackage{etoolbox}
\usepackage{mdframed} 	
\usepackage{xpatch}
\usepackage{multicol}
\usepackage[dvipsnames]{xcolor} 
\usepackage{hyperref}
\usepackage{doi}
\usepackage[backend=biber, sorting=none, style=numeric]{biblatex}

\hypersetup{
	colorlinks=true,
	linkcolor=MidnightBlue,
	citecolor=ForestGreen,
	urlcolor=RoyalBlue
}

\newtheoremstyle{plainroman} 
{}                          
{}                          
{\normalfont}               
{}                          
{\bfseries}                 
{.}                         
{.5em}                      
{}                          

\theoremstyle{plainroman}
\newtheorem{theorem}{Theorem}
\newtheorem{lemma}[theorem]{Lemma}     
\newtheorem{cor}[theorem]{Corollary} 
\newtheorem{prop}[theorem]{Proposition}

\newtheorem{definition}[theorem]{Definition}

\theoremstyle{definition}

\newcommand{\st}[1]{\star_{_{\mathcal{#1}}}}
\newcommand{\dt}[1]{\cdot_{_{\mathcal{#1}}}}

\DeclareMathOperator*{\argmax}{argmax}
\DeclareMathOperator*{\argmin}{argmin}

\newtcolorbox{dbox}[1][]{colback=green!5!white, colframe=white, boxrule=0pt, 
	fonttitle=\bfseries, #1}

\newtcolorbox{tbox}[1][]{colback=blue!5!white, colframe=white, boxrule=0pt, 
	fonttitle=\bfseries, #1}

\renewcommand{\eqref}[1]{Eq.~(\ref{#1})}
\renewcommand{\tr}[1]{\mathrm{Tr}[{#1}]}

\newcommand{\sumin}{\sum_{i=1}^n}
\newcommand{\sumid}{\sum_{i=1}^d}
\newcommand{\tiy}{\Theta^\dagger (I_{\mathcal Y})}
\newcommand{\cpxy}{\mathrm{CP} (\mathcal X, \mathcal Y)}
\newcommand{\pxy}{\mathbb P_{\mathcal X \otimes \mathcal Y}}
\newcommand{\xoy}{\mathcal X \otimes \mathcal Y}
\newcommand{\cpt}{\mathrm{CPT} (\mathcal X, \mathcal Y)}
\newcommand{\cpty}{\mathrm{CPT} (\mathcal Y, \mathcal X)}
\newcommand{\cphi}{\mathfrak{C}(\Phi)}
\newcommand{\cth}{\mathfrak{C}(\Theta)}

\newcommand{\fc}{\mathfrak{C}}
\newcommand{\fj}{\mathfrak{J}}

\newcommand{\h}{\mathcal{H}}
\newcommand{\x}{\mathcal{X}}
\newcommand{\y}{\mathcal{Y}}
\newcommand{\z}{\mathcal{Z}}
\newcommand{\inn}{{i \in [n]}}
\newcommand{\ot}{\otimes}
\newcommand{\ix}{I_\x}
\newcommand{\iy}{I_\y}

\newcommand{\ic}{\intercal}
\newcommand{\iv}{^{-1}}
\newcommand{\ihf}{^{-\frac12}}
\newcommand{\hf}{^{\frac12}}
\newcommand{\ra}{\rangle}
\newcommand{\la}{\langle}
\newcommand{\qaq}{\quad \text{and} \quad}
\newcommand{\mfr}[1]{\mathfrak{#1}}

\newcommand{\trz}[1]{\mathrm{Tr}_\z[{#1}]}
\newcommand{\kb}[1]{\ketbra{#1}{#1}}
\newcommand{\vc}{\mathrm{vec}}
\newcommand{\mm}{\mathbb{M}}
\newcommand{\mx}{\mathbb{M}_{\x}}
\newcommand{\mxy}{\mathbb{M}_{\x, \y}}
\newcommand{\dx}{\mathbb{D}_{\x}}
\newcommand{\px}{\mathbb{P}_{\x}}

\newcommand{\te}{\Theta}
\newcommand{\kp}{\mathrm{K}}
\newcommand{\kpr}{\mathrm{K}_{\rho}}

\newcommand{\ols}{\ostar_{\mathrm{LS}}}
\newcommand{\plc}{\Pi_{\Lambda, C}}
\newcommand{\dst}{\mathrm{Dist}}
\newcommand{\dsta}{\mathrm{Dist}_{\mathsf{A}}}
\newcommand{\dstb}{\mathrm{Dist}_{\mathsf{B}}}
\newcommand{\dstab}{\mathrm{Dist}_{\mathsf{AB}}}

\newcommand{\pab}{p_{\mathsf{AB}}}
\newcommand{\pa}{p_{\mathsf{A}}}

\newcommand{\mrm}[1]{\mathrm{#1}}
\newcommand{\msf}[1]{\mathsf{#1}}

\newcommand{\lic}{\Lambda^{-1}[C]}

\newcommand{\bb}[1]{\mathbb{#1}}

\newcommand{\cn}{\mathbb{C}^n}

\newcommand{\bbp}{\mathbb{P}}

\newcommand{\cl}[1]{\mathcal{#1}}

\newcommand{\bpq}{\mathrm{B}(P,Q)}
\newcommand{\fpq}{\mathrm{F}(P,Q)}

\let\oldtableofcontents\tableofcontents
\renewcommand{\tableofcontents}{{\hypersetup{linkcolor=black}
\oldtableofcontents}}

\begin{document}
	
	\title{\vspace*{-2cm}Projections with Respect to Bures Distance and Fidelity: Closed-Forms and Applications
}
	\date{}
	
	\author[1,2]{A. Afham\thanks{Email: \texttt{afham@nus.edu.sg}}}
	\author[1]{Marco Tomamichel}
	
	\affil[1]{\small Centre for Quantum Technologies, National University of Singapore, Singapore 117543, Singapore}
	\affil[2]{\small Centre for Quantum Software and Information, University of Technology Sydney, NSW 2007, Australia}

	\maketitle
    \vspace*{-0.5cm}
	\begin{abstract}
        We derive simple and unified closed-form expressions for projections with respect to fidelity (equivalently, the Bures and purified distances) onto several sets of interest. These include projections of bipartite positive semidefinite (PSD) matrices onto the set of PSD matrices with a given marginal, and projections of ensembles of PSD matrices onto the set of PSD decompositions of a given matrix, with important special cases corresponding to projections onto the set of quantum channels (via the Choi isomorphism) and onto the set of measurements.
        We introduce \emph{prior-channel decompositions} of completely positive (CP) maps, which uniquely decompose any CP map into a prior PSD matrix and a quantum channel. This decomposition generalizes the Choi--Jamiołkowski isomorphism by establishing a bijective correspondence between arbitrary bipartite PSD matrices and channel--state pairs, and we show that it arises naturally from the fidelity projections developed here.  
        
        As applications, we show that the \emph{pretty good measurement}—associated with a weighted ensemble—is the fidelity projection of the ensemble onto the set of measurements, and that the Petz recovery map—associated with a reference state and forward channel—is the projection of a CP map (constructed from the channel-state pair) onto the set of reverse quantum channels, thereby recasting the well-known identification of the Petz map with quantum Bayes’ rule in information-geometric terms. 
        Our results also provide an information-geometric underpinning of the Leifer--Spekkens quantum state over time formalism [\href{https://doi.org/10.1103/PhysRevA.88.052130}{Leifer and Spekkens, Phys. Rev. A 88, 052130 (2013)}].
        
	\end{abstract}
	
	
    \tableofcontents
	
	\section{Introduction}
	The \textit{projection problem}---the task of finding the closest point in a set to a given point~\cite{Boyd2004}---is of great importance in applied sciences.
	As quantum information involves various (sub)sets such as the sets of density matrices, quantum channels, separable states, and measurements, among others, various problems of interest in quantum information can be cast as projection problems or, more generally, constrained optimization problems.
	Some examples include quantum state discrimination~\cite{Watrous2018Theory, bae2015quantum}, separability measures~\cite{Vedral1997Quantifying, Vedral1998Entanglement, Streltsov2010Linking}, coherence measures~\cite{Baumgratz2014Quantifying}, and state and process tomography~\cite{smolin2012Efficient, guta2020fast, surawy2022projected}.
	
	Perhaps the earliest and most well-known projection problem is to project a point to a line with respect to the Euclidean distance.
	The solution, as the Pythagorean theorem would tell us, is obtained by dropping the perpendicular from the point onto the line. 
	
	Formally, we are given a set $\mathcal V$ which has a squared distance (or a divergence) $\mathrm{D}$ defined on it. 
	Then the \textit{projection of} $x \in \mathcal V$ to a (non-empty and compact) subset $ \mathcal S \subset \mathcal V$ is the (a) solution to the optimization problem
	\begin{equation}
		\Pi_{\mathcal S}^\mathrm{D} [x] := \argmin_{y \in \mathcal S} \mathrm{D}(y; x),
	\end{equation}
	
	The projection problem depends on the choice of the distance measure $\rm{D}$ and the \textit{feasible set} $\mathcal{S}$. 
	Typical choices for the distance measure include trace distance, Euclidean distance, KL divergence, and the distance induced by the spectral norm. 
	For feasible sets that are \textit{simple enough}, one can derive closed-forms for the projection~\cite{Boyd2004}. 
	However, for most feasible sets, even if they are convex, one typically has to resort to numerical methods~\cite{Liu2009Efficient, duchi2008efficient, Dhillon2008Matrix, Krichene2015Efficient}.
	
	By its relation to \textit{fidelity}, a distance which is crucial to quantum information is the Bures distance. 
	The squared Bures distance between positive semidefinite (PSD) matrices $P$ and $Q$ is defined as
	\begin{equation}
		\mathrm{B}(P,Q) := \mathrm{Tr}[P+Q] - 2 \mathrm{F}(P,Q),
	\end{equation} 
	where $\operatorname{F}(P, Q) := \left\| \sqrt Q \sqrt P \right\|_1 =\operatorname{Tr}\left[\sqrt{Q ^{\frac12} P Q ^{\frac12} }\right] $ is the (square root) \textit{fidelity} between $P$ and $Q$~\cite{Uhlmann2011Transition}.
	When restricted to density matrices, the squared Bures distance takes the form $\mrm{B}(\rho, \sigma) = 2[1 - \mrm{F}(\rho, \sigma)]$. 
	
	Attractive features of fidelity include monotonicity under quantum channels\footnote{Note that the Hilbert--Schmidt distance is \textit{not} monotonic under quantum channels~\cite{ozawa2000entanglement}.}~\cite{Watrous2018Theory, Uhlmann2010}, strong concavity~\cite{bhatia2018strong}, and operational interpretations. 
	The Bures distance is the natural distance of the (Riemannian) Bures manifold of positive definite matrices~\cite{bhatia2019bures, bengtsson2017geometry, Afham2025Riemanniangeometric}.  
	Fidelity also features in the definition of the \textit{purified distance}, a metric---on the set of subnormalized states---with applications in quantum cryptography~\cite{tomamichel2010duality, tomamichel2013thesis}.
	The celebrated Fuchs-van de Graaf inequalities~\cite{fuchs2002cryptographic} show that the squared Bures distance and purified distance lower- and upper-bound the trace distance.
	By its relation to the Wasserstein distance between Gaussian distributions, the \textit{Bures--Wasserstein} distance~\cite{bhatia2019bures} has also garnered considerable interest in machine learning (see~\cite{chewi2020gradient, altschuler2021averaging} and references therein). 
	Motivated by these features, in this article, we consider projections with respect to the Bures distance.

	\subsection{Problem Statement}
	Let $P \in \bbp_\h$, a PSD matrix acting on a finite-dimensional Hilbert space $\h$, be the point to project. 
	All the feasible sets we consider are specified by a pair $(\Lambda, C) \in \mrm{CPT}(\h,\mathcal K) \times \bbp_\mathcal K^+$, where $\Lambda$ is a quantum channel and $C$ is a positive definite matrix.\footnote{We refer to Section~\ref{Sec:Preliminaries} for notations.} 
	Formally, the feasible set is defined as
	\begin{equation}
		\Lambda^{-1}[C] := \{Q \in \bbp_\h : \Lambda(Q)=C\},
	\end{equation}
	that is, the PSD preimage of a fixed \emph{output matrix} $C$ under $\Lambda$.
	Such sets are called \textit{spectrahedra}~\cite{vinzant2014spectrahedron, chiribella2024extreme}. 
	Examples from quantum information include the sets of quantum states and quantum channels, and the collection of all POVMs.  
	
	\begin{table}[t]
		\centering
		\begin{tabular}{|c|c|p{90mm}|} 
			\hline
			\textbf{$(\Lambda, C)$} & Projection & Remarks \\
			\hline
			$(\operatorname{Tr}, c)$ & $ c \dfrac{P}{\operatorname{Tr}[P]}$ & $c \in \mathbb R_+$. $c=1$ yields trace normalization. \\
			$(\operatorname{Tr_\mathcal Y}, C)$ & $ P \st{X} (P_\mathcal X ^{-1} \# C)^2 $ & $P \in \mathbb P_{\mathcal X \otimes \mathcal Y}$. $C \in \mathbb P_\mathcal X$. $C = I_\mathcal X$ gives projection to Choi matrices of $\mathrm{CPT} (\mathcal X, \mathcal Y)$. \\
			$(\operatorname{Tr_\mathcal X}, C)$ & $ P \st{Y} (P_\mathcal Y ^{-1} \# C)^2 $ & $P \in \mathbb P_{\mathcal X \otimes \mathcal Y}$. $C \in \mathbb P_\mathcal Y$. $C = I_\mathcal Y$ gives projection to Choi matrices of $\mathrm{CPU} (\mathcal X, \mathcal Y)$. \\
			$(\operatorname{\Delta}, C)$ & $ \displaystyle \sum_{i,j=1}^{d} \sqrt{\frac{C_i C_j}{P_{ii} P_{jj}}} P_{ij} \ketbra{i}{j}  $& $\Delta$ is the completely dephasing map. $C = \sumid C_i \ketbra{i}{i}$ is a diagonal positive definite matrix. \\
			$(\operatorname{\Lambda}, C)$ & $ P \star \left(\displaystyle \bigoplus_{i=1}^{n} [P_i] ^{-1} \# C_i \right)^2 $ & $\Lambda$ is the pinching channel. $C := \oplus_{i=1}^n C_i $ is a block diagonal matrix compatible with $\Lambda$. $P_i$ are the corresponding blocks of $P$ (the form described on the left assumes block-diagonality in computational basis). \\
			$(\operatorname{M}, c)$ & $ P \star \left(\displaystyle \sum_{i=1}^{n} \sqrt{\dfrac{c_i}{\langle E_i, P\rangle} } E_i \right)^2 $ & $\mathrm{M}$ is a projective measurement with orthogonal projectors $\{E_i\}$. $c \in \mathbb R^n_+$ is a positive vector. \\ 
			$(\operatorname{Tr}_{\mathcal Y}, C)$ & $ P_i \star \left(P ^{-1} \# C\right)^2 $ & Corresponds to the projection of the ensemble $\{P_i\}_{i=1}^n$ to the set of decompositions of $C$. Here $P = \sum_{i=1}^{n} P_i$. $C = I_\mathcal X$ gives the projection to the set of $n$-outcome POVMs.
			\\
			\hline
		\end{tabular}
		\caption{Projection closed-form for specific channels. 
			Here $P$ is the point to be projected (with respect to fidelity / Bures distance) to the feasible set $\Lambda^{-1}[C] := \{Q \in \mathbb{P}_\mathcal H: \Lambda(Q) = C\}$, where $\Lambda \in \mrm{CPT}(\mathcal{H, K})$ and $C \in \mathbb P^+_\mathcal K$.
			We define $A \star B := B \hf A B \hf$ for $B \geq 0$ and $Z \st X X := Z \star (X \ot \iy)$ for $X \in \px$ and $Z \in \pxy$. 
			See Section~\ref{Sec:Preliminaries} for notations. 
		}
		
		\label{Tab:ClosedForms}
	\end{table}

	This framework encompasses several important examples.
	If $\Lambda=\Tr$ and $C\equiv c\in\mathbb R_+$, then $\mrm{Tr} \iv [c]$ is the set of PSD matrices with trace $c$, recovering the set of density matrices $\mathbb D_\h$ when $c=1$.
	If $\h \equiv \xoy$, $\mathcal K\equiv \x$, $\Lambda=\mrm{Tr}_\y$, and $C\in\bbp_{\x}$, then
	\begin{equation}
		\Lambda^{-1}[C] = \mrm{Tr}_\y \iv [C] =\{Q\in\pxy:Q_\x=C\}, 
	\end{equation}
	is the set of bipartite matrices with $\x$-marginal equal to $C$.
	When $C\equiv I_\x$, this set is Choi-isomorphic to the set of quantum channels $\cpt$.
	
	Thus, a \emph{Bures projection problem} is specified by a triple $(\Lambda,P,C)\in\mrm{CPT}(\h,\mathcal K)\times\bbp_\h\times\bbp_\mathcal K^+$
	and is defined as
	\begin{equation} \label{Eq:ProblemStatement}
		\Pi_{\Lambda, C} [P]
		:=
		\argmin_{Q \in \Lambda^{-1}[C]} \mathrm{B}(P,Q)
		=
		\argmax_{Q \in \Lambda^{-1}[C]} \operatorname F(P,Q),
	\end{equation}
	where $\Pi_{\Lambda,C}[P]$ denotes the set of \emph{Bures projections} of $P$ onto $\Lambda^{-1}[C]$.
	The equivalence between minimizing the squared Bures distance and maximizing fidelity follows from the expansion of $\mathrm B(P,Q) = \mrm{Tr}[P+Q] - 2 \fpq$ together with the fact that all feasible points satisfy
	$\Tr[Q]=\Tr[\Lambda(Q)]=\Tr[C]$.
	This equivalence may fail for other feasible sets.
	
	Moreover, if  the PSD matrices involved are restricted to subnormalized states $(\{S \in \bbp_d : \tr{S} \leq 1\})$, then our results also hold for projections with respect to the \textit{purified distance}~\cite{tomamichel2013thesis} whose squared version is defined as
	\begin{equation}
		\mrm{d}^2_{\mrm{Pur}}(P, Q) := 1 - \left(\fpq + \sqrt{(1-\tr{P}) (1-\tr{Q}))}\right)^2.   
	\end{equation}
	The fact that our results hold with respect to the purified distance follows from the observation that the trace terms are fixed for a given projection problem.
	
    We emphasize that, for all problems considered here, projections with respect to fidelity and Bures distance are equivalent. Accordingly, we refer to them interchangeably as fidelity (Bures) projections
	If the projection is unique, we slightly abuse notation and identify $\Pi_{\Lambda,C}[P]$ with this unique element.

	We now discuss the (non-)uniqueness of the projection.
	If $P$ is full-rank, the projection is unique by the strict concavity of fidelity~\cite{bhatia2018strong}.
	If $\Lambda(P)\perp C$, all feasible points are equidistant from $P$ and hence are (non-unique) projections.
	Nevertheless, uniqueness holds on the support of $P$: for any pair of projections $Q,Q'\in\Pi_{\Lambda,C}[P]$, we have $P^0 Q P^0 = P^0 Q' P^0$, where $P^0$ denotes the projector onto the support of $P$
	(see Theorem~\ref{Thm:ProjectionUniquenessCombined}).
	For simplicity, we henceforth assume that $P$ is full-rank.

	For any feasible point $Q\in\Lambda^{-1}[C]$, the data-processing inequality (DPI) for fidelity~\cite{Watrous2018Theory} implies
	\begin{equation}\label{Eq:DPIFid}
		\operatorname F(P,Q)\leq \operatorname F(\Lambda(P),\Lambda(Q))  = \operatorname F(\Lambda(P),C).
	\end{equation}
	We call the problem \emph{DPI-saturating} if the above inequality is saturated for some feasible $Q \in \lic$.
	Our main result shows that for such problems, the Bures projection admits a closed form via the \emph{Gamma map}
	\begin{equation}
		\Gamma_{\Lambda,C}[P] :=
		\bigl[\Lambda^\dagger(\Lambda(P)^{-1}\# C)\bigr]
		P
		\bigl[\Lambda^\dagger(\Lambda(P)^{-1}\# C)\bigr],
	\end{equation}
	where $A \# B := A \hf \sqrt{A \ihf B A \hf } A \ihf $ is the \textit{matrix geometric mean}~\cite{BhatiaPD} between positive definite matrices $A, B$. 
	Moreover, if $\Gamma_{\Lambda,C}[P]$ is feasible, then it is necessarily the Bures projection: $    \Gamma_{\Lambda,C}[P]\in\Lambda^{-1}[C] \Longrightarrow \Gamma_{\Lambda,C}[P]=\Pi_{\Lambda,C}[P].$
	The projection problems induced by partial trace, pinching channels, and projective measurements with mutually orthogonal projectors are all DPI-saturating and therefore admit closed-form solutions.

	\subsection{Examples}
	Through concrete examples, we now illustrate how Bures projections arise naturally in quantum information theory and how they have already appeared--often implicitly--in the literature.
	
	\medskip
	
	\paragraph{Normalization to density matrices.}
	We begin with the simplest case. 
	Let $\mathbb{D}_\mathcal{H}$, denoting the set of density matrices on
	$\mathcal{H}\equiv\mathbb C^d$, be the feasible set. 
	This corresponds to the choice $(\Lambda, C) = (\mrm{Tr}, 1)$. 
	Our results show that the Bures projection of an arbitrary PSD matrix
	$P\in\bbp_\mathcal H$ onto $\mathbb{D}_\mathcal{H}$ is simply the trace-normalization:
	\begin{equation}
		\Pi_{\mathbb{D}_\mathcal{H}}[P] \equiv \Pi_{\mrm{Tr}, 1}[P] = \frac{P}{\tr{P}}.
	\end{equation}
	While natural, it is important to emphasize that this operation generally does \emph{not} coincide with projections induced by other distances, such as the
	Hilbert--Schmidt distance.
	In the case of trace distance, the trace-normalization serves as a \textit{non-unique} projection.
	The coincidence between Bures projections and such canonical normalizations is a recurring theme.

	\paragraph{Fixed-marginal projections.}
	Next, consider the set $\mathcal S := \{ Q \in \bbp_{\mathcal{X}\otimes\mathcal{Y}} : Q_\mathcal X = C\}$, 
	where $Q_\mathcal X := \Tr_\mathcal Y[Q]$.
	This corresponds to the choice $(\Lambda, C) = (\mrm{Tr}_\y, C) \in \mrm{CPT}(\xoy, \x) \times \bbp_\x^+$.
	The Bures projection of $P\in\bbp_{\mathcal{X}\otimes\mathcal{Y}}^+$ onto $\mathcal S$ is
	\begin{equation}\label{Eq:MarginalProjIntro}
		\Pi_{\mathcal S}[P]  \equiv \Pi_{\mrm{Tr}_\y, C} = \bigl([P_\mathcal X]^{-1} \# C \otimes I_\mathcal Y\bigr)P \bigl([P_\mathcal X]^{-1} \# C \otimes I_\mathcal Y\bigr).
	\end{equation}
	When $C=I_\mathcal X$, this reduces to the common \textit{partial normalization} used to obtain a CPT map from a CP map under the Choi isomorphism.
	Our results show that this partial normalization is a bona fide projection to the set of CPT maps. 
	
	\paragraph{Ensemble projections and the PGM.}
	Finally, we consider the problem of \textit{ensemble projections}. 
	Let $\mathcal P := (P_i)_{i=1}^n\subset\bbp_\mathcal H$ be the ensemble we wish to project. 
	For $Q \in \bbp_\h^+$, let  
	\begin{equation}
		\mathrm{Dec}_n(Q) : \left\{(Q_i)_{i=1}^n \subset \bbp_\mathcal H : \sumin Q_i = Q\right\},
	\end{equation}
	be the set of all $n$-length decompositions of $Q$. 
	The Bures projection of $\mathcal P$ onto $\mathrm{Dec}_n(Q)$ is given by $\mathcal{Q}:= (Q_i)_{\inn}$ where
	\begin{equation}
		Q_i = (P^{-1}\#Q)\,P_i\,(P^{-1}\#Q),
		\qquad P := \sumin P_i .
	\end{equation}
	The squared Bures distance between two PSD ensembles $\cl P := (P_i)_\inn, \cl Q := (Q_i)_\inn$ of equal length is defined as $\mrm{B}(\mathcal{P,Q}) := \sumin \mrm{B}(P_i,Q_i)$. See Section~\ref{Sec:EnsProj} for further details. 
	For $Q=I_\mathcal H$, this yields the \emph{pretty good measurement}~\cite{hausladen1994pretty},
	providing a geometric interpretation of the PGM as a Bures projection.
	
	\medskip
	
	\subsection{Article Organization}
	We summarize the structure and the central results of this article.
	
	Section~\ref{Sec:Preliminaries} introduces notation and mathematical background.
	In Section~\ref{Sec:ClosedFormProjection} we present our main result: a unified closed-form expression for Bures projections (equivalently, fidelity maximization) of PSD matrices and ensembles onto convex, compact feasible sets defined as PSD preimages of a fixed \emph{output} matrix under a quantum channel, referred to as a \emph{constraint pair}.
	We derive explicit closed forms for several important channels—most notably partial trace—and summarize these formulas in Table~\ref{Tab:ClosedForms}.
	
	In Section~\ref{Sec:CPDecomp} we introduce \emph{prior-channel decompositions}, which uniquely decompose any completely positive map into an unnormalized \emph{prior} PSD matrix and a normalized quantum channel.
	This generalizes the fact that (non-zero) PSD matrices admit a unique decomposition into pairs of a scalar and a density matrix. 
	Moreover, the decomposition admits a geometric interpretation in terms of Bures projections at the level of Choi matrices (and Frobenius projections at the level of Stinespring representations).
	This formalism is used extensively in later sections.
	
	Section~\ref{Sec:ApplicationsAndManifestations} discusses applications of our closed-forms for fidelity projections.
	These include closed-form projections to various sets of interest, geometric interpretations of the Petz recovery map and the pretty good measurement, information geometry of the Leifer--Spekkens state-over-time formalism, a simplified proof and geometric understanding of the quantum minimal change principle, and geometric interpretations of several channel distance measures.
	
	Finally, Section~\ref{Sec:Conclusion} concludes with open problems and directions for future research.

	\subsection{Related Works}
	\textcite{brahmachari2025fixed} study Bures projections onto the set of \textit{invariant PSD matrices} with respect to a group symmetry. 
	Note that they find a fixed-point iteration algorithm (which can be seen as Riemannian gradient descent on Bures manifold~\cite{Afham2025Thesis}) whereas we are interested in closed forms. 
	\textcite{junyi2024convergence} find closed form for projections onto \textit{Bures-(Wasserstein) balls}, where the feasible sets are of the form $\mathcal{S} \equiv \mathcal{S}(Q, r):= \{Q': \mrm{B}(Q,Q') \leq r^2\}$ for a PSD $Q$ and a \textit{radius} $r>0$.  
	\textcite{Afham2022Quantum} study the problem of average fidelity maximization, which can be seen as projected gradient descent on the Bures manifold~\cite{Afham2025Thesis}.
	
	\section{Mathematical Preliminaries} \label{Sec:Preliminaries}
	We introduce the notations and some helpful mathematical preliminaries.
	
	\subsection{Notations}
	For a positive integer $n$, we define the index set $[n] := \{1, \ldots, n\}$.
	For a complex number $z$, $\Re(z)$ denotes its real part. 
	For a complex Hilbert space $\mathcal X = \mathbb{C}^d$, we will use $\mathbb M_\mathcal X, \mathbb H_\mathcal X, \mathbb P_\mathcal X, \mathbb P_\mathcal X^+, \mathbb U_\mathcal X$, and $\mathbb D_\mathcal X $ to denote the set of all square, Hermitian, positive semidefinite (PSD), positive definite (PD), unitary, and density matrices acting on $\mathcal X$.
	The above objects would be denoted by $\mathbb M_d, \mathbb H_d, \mathbb P_d, \mathbb P_d^+, \mathbb U_d$, and $\mathbb D_d$ respectively if we intend to highlight the dimension $d$ over the label $\mathcal X$.
	The identity matrix on $\x$ is denoted as $I_\x$.
	The set of matrices corresponding to linear maps from $\x$ to $\y$ is denoted by $\mathbb{M}_{\x, \y}$ and $\mathbb{U}_{\x, \y}$ denotes the set of isometries from $\x$ to $\y$: $\mathbb{U}_{\x, \y} := \{V \in \mxy : V^\dagger V = \ix \}$.
	The kernel (null space) of an operator $A \in \mathbb M_{\x, \y}$ is defined as $\mathrm{ker}(A) := \{|u \ra  \in \x : A |u\ra  = 0\}$.
	The \textit{support} of an operator is the orthogonal complement of its kernel, and is denoted as $\mathrm{supp}(A)$. 
	For operators $A,B \in \mathbb{H}_\h$, we use $A \ll B$ to denote $\mathrm{supp}(A) \subseteq \mathrm{supp}(B)$, and $A^0$ to denote the orthogonal projector onto $\mrm{supp}(A)$.  
	
	We use the notation $A \geq 0~ (A > 0)$ to convey that $A$ is PSD (PD). 
	For any $A \geq 0$, we use the equivalent notations $A^{\frac12}$ and $\sqrt{A}$ to denote the (PSD) square root of $A$. 
	The Hilbert--Schmidt (Frobenius) inner product on $\mx$ is defined as $\la A, B \ra := \tr{A^\dagger B}$, and the induced norm is denoted as $\|A\|_\mrm{F} := \sqrt{\la A,A \ra}$.
	The Euclidean norm (of vectors) is denoted as $\|u\|_2 := \sqrt{\la u|u \ra} \equiv \sqrt{\la u,u \ra}$, where $\la \cdot, \cdot \ra $ is the standard inner product over $\mathbb C^d$.
	For a function $f :  \mathcal{A} \to \cal B$, and any $\cal C \subseteq \mathcal{A}$, we denote $f[\mathcal C] \equiv \{f(c): c \in \mathcal C\} \subseteq \mathcal{B}$.

	We denote the set of all linear, Hermitian-preserving, completely positive (CP), completely positive and trace-preserving (CPT, a.k.a. quantum channels), and completely positive unital maps from $\mathbb M_\mathcal H$ to $\mathbb M_\mathcal K$ as $\mathrm{LM}(\mathcal H, \mathcal K),$ $\mathrm{HP}(\mathcal H, \mathcal K),$ $\mathrm{CP}(\mathcal H, \mathcal K),$ $\mathrm{CPT}(\mathcal H, \mathcal K),$ and $\mathrm{CPU}(\mathcal H, \mathcal K)$ respectively.
	The (Hilbert--Schmidt) adjoint of a linear map $\te \in \mrm{LM}(\cal H, \cal K)$ is the unique linear map $\te^\dagger  \in \mrm{LM}(\cal K, \cal H)$ satisfying $\la K, \Theta(H) \ra = \la \te^\dagger (K), H \ra $ for all $H \in \mathbb M_{\cal H}$ and $K \in \mathbb M_{\cal K}$.
	The identity map on $\mathbb M_\mathcal{H}$ is denoted as $\mrm{Id}_\mathcal{H} \in \mrm{LM}(\mathcal{H,H})$. 
	
	The \textit{Choi isomorphism} bijectively associates linear maps $\mathrm{LM}(\mathcal H, \mathcal K)$ and bipartite matrices $\mathbb M_{\mathcal H \otimes \mathcal K}$. 
	For a linear map $\Phi \in \mathrm{LM}(\mathcal H, \mathcal K)$ and a bipartite matrix $M \in \mathbb M_{\mathcal H \otimes \mathcal K}$, the Choi isomorphism $\mathfrak{C} : \mathrm{LM} (\mathcal H, \mathcal K) \to \mathbb M_{\mathcal H \otimes \mathcal K} $ and its inverse are defined as
	\begin{equation} \label{Eq:ChoiIsoDef}
		\mathfrak{C}(\Phi) := (\mathrm{Id}_\mathcal H \otimes \Phi)(\Omega) = \sum_{i,j=1}^{\mathrm{dim}(\mathcal H)} \ketbra{i}{j} \otimes \Phi(\ketbra{i}{j}) \quad \text{and} \quad \mathfrak{C}^{-1}(M)(H) = [M (H^\intercal \otimes I_\mathcal K)]_\mathcal K,
	\end{equation}
	for any $H \in \mathbb M_\mathcal{H}$, where $\Omega = \kb{\omega}$ and $|\omega\ra = \sum_{i=1}^{\mrm{dim}(\x)} |i,i\rangle$.
	
	For a Hermitian-preserving map $\Theta \in \mrm{HP}(\h, \mathcal K)$, a simple relation holds between its Choi matrix and that of its adjoint: $\fc(\Theta^\dagger ) = \fc (\Theta)^\ic$. 
	See Appendix~\ref{App:ChoiMatrixAdjoint} for a proof. 
	The relation has been noted previously for CPT~\cite{Johnston2011Quantum} and CP maps~\cite{Stormer2011Mapping}.
	It follows that $\Theta^\dagger (Y) = [\fc(\Theta)^\ic (I_\x \otimes Y)]_\x$.

	We now recall the \textit{star product}---a non-commutative, non-associative product over matrices.
	This notation is closely related to the \textit{Leifer-Spekkens state over time} representation used in~\cite{Leifer2013towards, leifer2006quantum}. 
	For arbitrary $A$ and $B \geq 0$, we define
	\begin{equation}
		A \star B :=  B ^{\frac12} A B ^{\frac12}. 
	\end{equation}
	Two useful identities regarding the star product are \textit{positivity}: $ A \star B \geq 0$ for $A, B \geq 0$ and the \textit{transpose property}: $(A \star B)^\intercal = A^\intercal \star B^\intercal$.
	
	The star product of more than two PSD matrices, when denoted without parentheses, is evaluated from left to right: $A \star B \star C := (A \star B) \star C = C ^{\frac12} B ^{\frac12} A B ^{\frac12} C ^{\frac12}$.
	Frequently, we take the star product with respect to a single subsystem of a multipartite system.
	Let $\mathcal Z:= \mathcal X_1 \otimes \cdots \otimes \mathcal X_n$, $A \in \mathbb M_{\mathcal Z}$, and $K \in \mathbb P_{{\mathcal X}_k}$ for some $k \in [n]$.
	We denote
	\begin{equation}
		A \star_{\mathcal X_k} K := A \star ({I}_{{\mathcal X}_1 \otimes \cdots \otimes {\mathcal X}_{k-1}} \otimes K \otimes {I}_{{\mathcal X}_{k+1} \otimes \cdots \otimes {\mathcal X}_{n}}).
	\end{equation}
	Essentially we \textit{star-multiply} $K$ to the $\mathcal X_k$-component of $A$.
	For bipartite matrices, this reads:
	\begin{equation}
		A \st{X} X := A \star (X \otimes I_\mathcal Y ) \quad \text{and} \quad 
		A \st{Y} Y := A \star (I_\mathcal X \otimes Y),
	\end{equation}
	for $A \in \mathbb M_{\mathcal X \otimes \mathcal Y }, X \in \mathbb P_\mathcal X,$ and $Y \in \mathbb P_\mathcal Y$. 
	We extend the above notation to the standard matrix product: 
	\begin{equation}
		A \dt{X} X := A \cdot (X \otimes I_\mathcal Y ) \quad \text{and} \quad 
		A \dt{Y} Y := A \cdot (I_\mathcal X \otimes Y).
	\end{equation}
	
	For $P \in \mathbb M_{\mathcal X_1 \otimes \cdots \otimes \mathcal X_n}$, we use $P_{\mathcal X_k}$ to denote the marginal of $P$ on the space $\mathcal X_k$, which is obtained by tracing out every other subsystem.
	We will also encounter expressions like $[P]_{\mathcal{X}_k}^{-1}$, which denote the inverse of a marginal.
	Often we omit the brackets and write this as $P_{\mathcal{X}_k}^{-1}$.
	In such cases, the marginalization is always performed before the inversion. 
	Observe that
	\begin{equation} \label{Eq:StarProdMarginalMulti}
		[A \star_{\mathcal X_k} K]_{\mathcal X_k} = A_{\mathcal X_k} \star K,
	\end{equation}
	for any $K \in \mathbb P_{\mathcal X _k}$.
	For bipartite systems, this simplifies to
	\begin{equation} \label{Eq:StarProdIden}
		\begin{aligned}
			[A \st{X} X]_\mathcal X := \mathrm{Tr}_\mathcal Y[A \st{X} X] = A_\mathcal X \star X  \quad \text{and} \quad  
			[A \st{Y} Y]_\mathcal Y  := \mathrm{Tr}_\mathcal X[A \st{Y} Y ] = A_\mathcal Y \star Y.
		\end{aligned}
	\end{equation}
	We define the symmetrized division as $\frac{A}{B} := A \star B ^{-1} =  B ^{-\frac12} A B ^{-\frac12}$, where negative powers are taken on the support.
	
	{The majority of the projections in this article are with respect to fidelity (Bures distance), and thus we reserve $\Pi$ (without any superscript) to denote such projections. 
		For a set $\cl S$, we use $\Pi_{\cl S}$ to denote the (fidelity/Bures) projection to $\mathcal{S}$, and for a constraint pair $(\Lambda, C)$, we use $\plc$ to denote the projection to the set $\lic$. 
		We also consider projections with respect to the \textit{Choi--Bures} distance (see Definition~\ref{Def:CBDistance}), which we denote by $\Pi^{\mrm{CB}}$. 
		Rarely, we consider Euclidean projections of vectors and Frobenius projections of arbitrary matrices, which are denoted by $\Pi^\mrm{Eu}$ and $\Pi^\mrm{F}$ respectively.} 
	
	\subsection{The Vec Isomorphism, Purifications, and Stinespring Representations}
	The vec mapping is a (basis-dependent) linear bijection between linear operators and bipartite vectors. 
	We offer a brief discussion here and refer to~\cite[Section 1.1.2]{Watrous2018Theory} and~\cite{Mele2024Introduction} for more detailed treatments.
	Consider the spaces $\mathbb{M}_{\x, \y}$ and $\x \ot \y$ and fix bases $\{|i \ra \}_{i=1}^{\mrm{dim}(\x)}$ and $\{|j \ra \}_{j=1}^{\mrm{dim}(\y)}$ of $\x$ and $\y$ respectively.
	The $\mrm{vec}$ mapping offers a bijection between these spaces, which is defined on standard basis vectors as
	\begin{equation}
		\vc(| j \ra  \la i|) := |i \ra  \ot |j \ra ,
	\end{equation}
	which is extended linearly to arbitrary operators and bipartite vectors. In particular, for arbitrary vectors $|x \ra  \in \x$ and $|y \ra  \in \y$, we have $\vc(|y \ra  \la x| ) = |\bar x \ra   \ot |y \ra $,
	where $|\bar x \ra := \sumin \bar x_i |i \ra $ is the complex conjugate of $|x \ra$. 
	
	This is slightly different from the definitions in~\cite{Watrous2018Theory, Mele2024Introduction} where the mapping is defined as $\vc( |y \ra  \la x| ) :=  |y \ra  \otimes |\bar x \ra $, and thus sending matrices from $\mathbb{M}_{\x, \y}$ to vectors in $\y \ot \x$.\footnote{In terms of tensor network diagrams, the convention followed by~\cite{Watrous2018Theory, Mele2024Introduction} can be thought of as the \textit{wire bending downwards} (see \cite[Eq. 140]{Mele2024Introduction}), whereas we \textit{bend the wire upwards}.}
	
	The vec mapping is an isometry in the following respect:
	\begin{equation}
		\la A, B \ra := \tr{A^\dagger  B} = \la \vc(A), \vc(B) \ra,  
	\end{equation}
	which further implies $\|A - B\|_{\mrm{F}} = \|\vc(A) - \vc(B) \|_2$.
	For any triple $A \in \mathbb{M}_{\x, \y}, X \in \mathbb{M}_{\x}$ and $Y \in \mm_{\y}$, it holds that
	\begin{equation}
		\vc(YAX) = (X^{\ic} \ot Y) \vc(A), 
	\end{equation}
	
	We now discuss \textit{purifications} of PSD matrices and the \textit{Stinespring representations} of CP maps.
	\begin{dbox}
		\begin{definition}[Purifications]
			Let $P \in \mathbb P_\mathcal X$. 
			A vector $|u \ra  \in \mathcal Z \otimes \mathcal X$ is a \emph{purification} of $P$ if $\operatorname{Tr}_{\mathcal Z}[|u \ra  \la u|] = P$.
			The set of all purifications of $P$, with the \emph{purification space} $\mathcal Z$, is
			\begin{equation}
				\mathrm{Pur}_{\mathcal Z}(P) := \{|u \ra  \in \mathcal Z \otimes \mathcal X : \mathrm{Tr}_\mathcal Z[\kb{u}] = P\}. 
			\end{equation}
		\end{definition}
	\end{dbox}	
	For a vector $|u \ra  \in \mathcal Z \otimes \mathcal X$ to be a purification of $P$, it must hold that $\mathrm{dim}(\mathcal Z) \geq \mathrm{rank}(P)$~\cite[Theorem 2.10]{Watrous2018Theory}.
	Moreover, the \textit{unitary equivalence of purifications}~\cite[Theorem 2.12]{Watrous2018Theory} states that $|v \ra , |v' \ra  \in \mathrm{Pur}_\mathcal Z(P)$ if and only if $|v \ra  = (U \otimes I_\mathcal X)|v' \ra $ for some unitary $U \in \mathbb U_\mathcal Z$.    
	The \textit{canonical} (standard) purification of $P \in \px$ is defined as
	\begin{equation}
		\vc(P \hf) = (I_\z \ot P \hf ) |\omega \ra \equiv |\omega_P \ra \in \z \ot \x, 
	\end{equation}
	where $\omega =  \sum_i |i,i \ra \in \z \ot \x$ is the unnormalized (canonical) maximally entangled state and $\z \cong \x$. 
	By unitary equivalence of purifications, every purification of $P$ can be written $(U \ot I_\x) |\omega_P \ra  $. 
	We will also use the notations $\Omega \equiv \kb{\omega}$ and $\Omega_P \equiv \kb{\omega_P}$ for any $P \in \px$.
	
	The analogue of purifications for linear maps (such as the elements of $\mrm{LM}(\x, \y)$) are \textit{Stinespring representations}. 
	We will be interested in the Stinespring representation of CP maps~\cite{Stinespring1955Positive, Watrous2018Theory}. 
	\begin{dbox}
		\begin{definition}
			Let $\Theta \in \mathrm{CP} (\mathcal X, \mathcal Y)$. 
			An operator $K \in \mathbb M_{\mathcal X, \mathcal Y \otimes \mathcal Z}$ is called a \emph{Stinespring operator}~\cite{Watrous2018Theory} of $\Theta$ if 
			\begin{equation}
				\Theta(X) = \mathrm{Tr}_\mathcal Z[KXK^\dagger]
			\end{equation}
			for all $X \in \mathbb M_\mathcal X$. 
			We denote the set of all Stinespring operators of $\Theta$ with the auxiliary space $\mathcal Z$ as 
			\begin{equation}
				\mathrm{SR}_\mathcal Z(\Theta) := \{K \in \mathbb M_{\mathcal X, \mathcal Y \otimes \mathcal Z} : \mathrm{Tr}_\mathcal Z[K \cdot K^\dagger] = \Theta(\cdot)\}.
			\end{equation}
		\end{definition}
	\end{dbox}
	For an operator $K \in \mathbb M_{\mathcal X, \mathcal Y \otimes \mathcal Z}$ to be a Stinespring representation of $\Theta$, it must hold that $\mathrm{dim}(\mathcal Z) \geq \text{rank}(\mathfrak{C}(\Theta))$~\cite[Theorem 2.22]{Watrous2018Theory}.
	As with purifications, Stinespring representations are not unique, and are equivalent up to a unitary in the auxiliary space $\z$.
	Formally, $K, K' \in \mathrm{SR}_\mathcal Z(\Theta)$ if and only if $K = (I_{\mathcal Y} \otimes U) K'$ for some unitary $U \in \mathbb U_\mathcal Z$~\cite[Corollary 2.24]{Watrous2018Theory}.   
	A CP map $\Theta$ is trace-preserving if and only if its Stinespring representations are isometries: $K^\dagger K = I_\x$ for all $K \in \mrm{SR}_\z(\Theta)$.
	If $\Theta$ is trace-preserving only on a subspace $\x ' \subseteq \x$, then $K$ will be a partial isometry with $K^\dagger K$ being a projector onto $\x'$.
	
	\subsection{The Choi--Bures and Stinespring--Frobenius Distances between Completely Positive Maps}
	We collect several known results relating the Bures distance between positive semidefinite (PSD) matrices to Euclidean distances between their purifications, together with their extensions to completely positive (CP) maps via the Choi--Jamiołkowski isomorphism and Stinespring representations.  
	These results allow one to interpret Bures-type distances as minimal Euclidean or Frobenius (Hilbert--Schmidt) distances over appropriate purification spaces.
	
	\medskip
	
	We begin with the matrix-valued setting.  
	The following result is a direct consequence of Uhlmann's theorem~\cite{uhlmann1976transition, Watrous2018Theory} and characterizes the Bures distance between two PSD matrices as the minimal Euclidean distance between their purifications.
	
	\begin{tbox}
		\begin{prop}\label{Prop:PurificationStates}
			Let $P, Q \in \mathbb P_\mathcal X$ and let $\mathcal Z \equiv \mathbb C^r$ with
			$r := \max\{\mathrm{rank}(P), \mathrm{rank}(Q)\}$.
			Then
			\begin{equation}
				\mathrm{B}(P,Q)
				= \min_{U \in \mathbb U_\mathcal Z}
				\| \ket p - (U \otimes I_\mathcal X )  |q \ra  \|_2^2	
				= \min_{\substack{|p' \ra  \in \mathrm{Pur}_\mathcal Z(P)\\ |q'\ra  \in \mathrm{Pur}_\mathcal Z(Q)}}
				\||p'\ra  - |q'\ra \|_2^2,
			\end{equation}
			for any fixed purifications $|p\ra  \in \mathrm{Pur}_\mathcal Z(P)$ and
			$|q\ra  \in \mathrm{Pur}_\mathcal Z(Q)$.
		\end{prop}
	\end{tbox}
	
	\begin{proof}
		The second equality follows from the unitary invariance of Euclidean norm and the unitary equivalence of purifications~\cite[Theorem~2.12]{Watrous2018Theory}.
		For the first equality, we compute
		\begin{equation}
			\begin{aligned}
				\min_{U \in \mathbb U_\mathcal Z}
				\| |p\ra  - (U \otimes I_\mathcal X )|q\ra  \|_2^2
				&=
				\|p\|_2^2 + \|q\|_2^2
				- 2 \max_{U \in \mathbb U_\mathcal Z}
				\Re \langle p, (U \otimes I_\mathcal X ) q \rangle \\
				&=
				\operatorname{Tr}[P+Q] - 2 \mathrm{F}(P,Q) =: \bpq,
			\end{aligned}
		\end{equation}
		where the maximization of the inner product is given by Uhlmann's theorem~\cite{uhlmann1976transition, Watrous2018Theory}.
	\end{proof}
	
	\medskip
	
	We now extend this viewpoint to CP maps.
	Since the Choi isomorphism maps CP maps to PSD matrices, we can naturally combine it with the definition of Bures distance to define the \textit{Choi--Bures distance}.  
	\begin{dbox}
		\begin{definition}[Choi--Bures distance] \label{Def:CBDistance}
			Let $\Theta, \Psi \in \cpxy$ be arbitrary CP maps. The \emph{Choi--Bures} distance is defined as 
			\begin{equation}
				\mathrm{d}_\mathrm{CB}(\Theta,\Xi)
				:=
				\sqrt{\mathrm{B}(\mathfrak C(\Theta), \mathfrak C(\Xi))}.
			\end{equation}
		\end{definition}
	\end{dbox}
	For quantum channels, related notions have appeared as operational fidelity~\cite{Belavkin2005Operational, Kretschmann2008Informationdisturbance}
	and entanglement fidelity~\cite{beny2010general}.
	
	As the Stinespring representation generalizes the notion of purifications of PSD matrices to CP maps, one can define a distance between CP maps as the minimal Hilbert--Schmidt (Frobenius) distance over their Stinespring representations as follows.
	\begin{dbox}
		\begin{definition}[Stinespring--Frobenius distance]
			Let $\Theta, \Xi \in \cpxy$ be arbitrary CP maps and let $\z \equiv \mathbb{C}^r$ with $r := \max \{\mrm{rank}(\fc(\Theta), \mrm{rank (\fc(\Xi))}\}$.
			We define the \emph{Stinespring--Frobenius} distance between $\Theta$ and $\Xi$ as
			\begin{equation}
				\mrm{d}_{\mrm{SF}}(\Theta, \Xi) := \min_{\substack{K \in \mathrm{SR}_\mathcal Z(\Theta)\\ L \in \mathrm{SR}_\mathcal Z(\Xi)}} \|K - L\|_\mrm{F} = 
				\min_{U \in \mathbb U_\mathcal Z}
				\| K' - (I_\mathcal Y \otimes U) L' \|_{\mrm{F}},
			\end{equation}
			where $K'$ and $L'$ are arbitrarily fixed Stinespring representations of $\Theta$ and $\Xi$ respectively.
		\end{definition}
	\end{dbox}
	Note that the second equality follows from the unitary invariance of the Frobenius norm and the unitary equivalence of Stinespring representations. 
	The following proposition, which can be seen as the generalization of Proposition~\ref{Prop:PurificationStates} to CP maps, shows that the Choi--Bures and Stinespring--Frobenius distances coincide.
	
	\begin{tbox}
		\begin{prop}\label{Prop:PurificationChannels}
			Let $\Theta, \Xi \in \mathrm{CP}(\mathcal X,\mathcal Y)$, $r := \max \{\mathrm{rank}(\mathfrak C(\Theta)), \mathrm{rank}(\mathfrak C(\Xi))\}$, and
			$\mathcal Z \equiv \mathbb C^r$.
			It holds that
			\begin{equation}
				\mrm{d}^2_{\mrm{CB}} (\Theta, \Xi) := \mathrm{B}(\mathfrak C(\Theta), \mathfrak C(\Xi))
				=
				\min_{\substack{K \in \mathrm{SR}_\mathcal Z(\Theta)\\ L \in \mathrm{SR}_\mathcal Z(\Xi)}}
				\|K - L\|_{\mrm{F}}^2 =: \mrm{d}^2_{\mrm{SF}}  (\Theta, \Xi).
			\end{equation}
		\end{prop}
	\end{tbox}
	
	\begin{proof}
		The equality will be proved by showing that any $K \in \mrm{SR}_\z(\Theta)$ and $L \in  \mrm{SR}_\z(\Xi)$ are $\mrm{vec}$-isomorphic to some purifications of $\fc(\Theta)$ and $\fc(\Xi)$ respectively.
		Begin with the fact that
		\begin{equation}
			\min_{K \in \mathrm{SR}_\mathcal Z(\Theta), L \in \mathrm{SR}_\mathcal Z(\Xi)} \| K - L\|_2^2 = \min_{U \in  \mathbb{U}_\mathcal Z} \| K - (I_\mathcal Y \otimes U) L\|_2^2, 
		\end{equation} 
		where we have arbitrarily fixed the Stinespring representations $K$ and $L$ of $\Theta$ and $\Xi$ respectively. 
		Expanding the RHS, we have
		\begin{equation}
			\min_{U \in  \mathbb{U}_\mathcal Z} \| K - (I_\mathcal Y \otimes U) L\|_2^2  =  \operatorname{Tr}[K^{\dagger}K + L^{\dagger}L] - 2 \max_{U \in \mathbb U_{\mathcal Z}} \Re \langle K, (I_\mathcal Y \otimes U) L \rangle .
		\end{equation}
		Since $\operatorname{Tr}[K^{\dagger}K] = \operatorname{Tr}[\mathfrak{C}(\Theta)]$ and $\operatorname{Tr}[L^{\dagger}L] = \operatorname{Tr}[\mathfrak{C}(\Xi)]$, it suffices show that the maximization of the inner product over the unitaries in the purification space is attained at $\operatorname{F}(\mathfrak{C}(\Theta), \mathfrak{C}(\Xi))$. 
		
		Any Stinespring representation $K \in \mathrm{SR}_\mathcal Z(\Theta)$ can be written as $K = \sum_{i=1}^{r} K_i \otimes |i \rangle$, where $\{K_i\}_{i \in [r]} \subset \mathbb M_{\mathcal X, \mathcal Y }$ form a Kraus representation of $\Theta$. Here $K_i := ({I_\mathcal Y \otimes \langle i|})K$ for each $i \in [r]$~\cite[Proposition 2.20]{Watrous2018Theory}. 
		By the $\mathrm{vec}$ isomorphism, we have 
		\begin{equation}
			\mathrm{vec}(K) = \mathrm{vec}\left(\sum_{i=1}^{n} K_i \otimes |i \rangle \right) =  \sum_{i=1}^{n} \mathrm{vec}\left(K_i \right) \otimes |i \rangle, 
		\end{equation}
		where we used the convention $\mathrm{vec}(|a \rangle \langle b| \otimes \ketbra{i}{j}) \equiv  \mathrm{vec}(|a \rangle \langle b|) \otimes \mathrm{vec}(|i \rangle \langle j|) = |b, a\rangle |j,i \rangle$. 
		Observe that $\mathrm{vec}(K) \in \mathrm{Pur}_\mathcal Z(\mathfrak{C}(\Theta))$, as
		\begin{equation}
			\mathrm{Tr}_\mathcal Z[\mathrm{vec}(K) \mathrm{vec}(K)^{\dagger}] = \sum_{i=1}^{n} \mathrm{vec}(K_i) \mathrm{vec}(K_i)^{\dagger} = \sum_{i=1}^{n} (\mathrm{Id_\mathcal X  } \otimes K_i)|\omega \rangle \langle \omega|(\mathrm{Id}_\mathcal X \otimes K_i^{\dagger}) = \mathfrak{C}(\Theta).
		\end{equation}
		In the second equality, we used the fact that $\mathrm{vec}(K_i) = (I_\mathcal X \otimes K_i)(|\omega \rangle )$. 
		From the relation $\langle A, B \rangle = \langle \mathrm{vec}(A), \mathrm{vec}(B) \rangle $, we have
		\begin{equation}
			\max_{U \in \mathbb U_\mathcal Z} \Re \langle K, (I_\mathcal Y \otimes U) L \rangle = \max_{U \in \mathbb U_\mathcal Z}  \langle \mathrm{vec}(K), (I_\mathcal Y\otimes U) \mathrm{vec}(L) \rangle = \operatorname{F}(\mathfrak{C}(\Theta), \mathfrak{C}(\Xi)).
		\end{equation}
		This equality, along with the equalities $\operatorname{Tr}[K^{\dagger}K] = \operatorname{Tr}[\mathfrak{C}(\Theta)]$ and $\operatorname{Tr}[L^{\dagger}L] = \operatorname{Tr}[\mathfrak{C}(\Xi)]$, allows us to conclude the proof. 
		
	\end{proof}
	Thus, the Bures distance between Choi matrices can be interpreted as the minimal Hilbert--Schmidt distance over Stinespring representations.
	A variant of this equivalence for quantum channels was previously derived in~\cite{Belavkin2005Operational}.
	
	Finally, combining Bures projections with the Choi--Bures distance yields the notion of
	\emph{Choi--Bures projections}.
	\begin{dbox}
		\begin{definition}[Choi--Bures projections] \label{Def:CBProj}
			Let $\Theta \in \mathrm{CP}(\mathcal X,\mathcal Y)$ and $\mathcal S \subseteq \mathrm{CP}(\mathcal X,\mathcal Y)$. 
			The \emph{Choi--Bures} projection of $\Theta$ to $\mathcal{S}$ is defined as
			\begin{equation}\label{Eq:CBProjDef}
				\Pi^{\mrm{CB}}_{\mathcal S}[\Theta]
				:=
				\Pi_{\mathfrak C[\mathcal S]}[\mathfrak C(\Theta)].
			\end{equation}    
		\end{definition}
	\end{dbox}
	These projections will play a central role in the remainder of this article, particularly when $\mathcal S$ is chosen to be a structured subset of CP maps, such as the set of quantum channels.
	
	\subsection{Matrix Geometric Mean}
	The \textit{matrix geometric mean}~\cite{Pusz1975Functional, kubo1980means, BhatiaPD} of $A, B \in \mathbb{P}_d^+$ is defined as
	\begin{equation} \label{Eq:DefMGM}
		A \# B := A ^{\frac12} \sqrt{A ^{-\frac12} B A ^{-\frac12} } A ^{\frac12} > 0. 
	\end{equation}  
	We see that if $AB = BA \Longrightarrow A \# B = A ^{\frac12} B ^{\frac12} = \sqrt{A B}$.
	The matrix geometric mean satisfies a variety of useful properties, some of which we list below.
	See~\cite{BhatiaPD, lawson2024expanding, Afham2025Riemanniangeometric} for proofs.
	\begin{tbox}
		\begin{prop}[Properties of the Matrix Geometric mean] \label{Prop:MatGeoMean} 
			The following statements hold true for any $A, B \in \mathbb{P}_d^+$. 
			\begin{enumerate}
				\item $A \# B = B \# A$ and $A ^{-1} \# B ^{-1} = (A \# B) ^{-1}$.
				\item $X = A ^{-1} \# B$ is the unique positive definite solution to the matrix Riccati equation $B = X A X$. 
				It follows that $A = X ^{-1} B  X ^{-1}$. 
			\end{enumerate}	
		\end{prop}
	\end{tbox}
	The gradient of fidelity (over full-rank states) can be expressed in terms of the geometric mean~\cite{bhatia2019bures}
	\begin{equation}
		\nabla_P \fpq = \frac12 P \iv \# Q.
	\end{equation}

	\subsection{Data Processing Inequality for Fidelity}
	Fidelity satisfies the \textit{data processing inequality}~\cite{Watrous2018Theory, Wilde2013, Uhlmann2011Transition}.
	For $P, Q \in \mathbb P_\mathcal H$ and any $\Lambda \in \mathrm{CPT} (\mathcal H, \mathcal K)$, it holds that $	 \operatorname{F}(P, Q) \leq \operatorname{F}(\Lambda(P), \Lambda(Q))$.
	We say that the triple $(\Lambda, P, Q)$ \textit{saturates} the DPI if equality is achieved in the above inequality. 
	\textcite{leditzky2017data} (also see~\cite{cree2022geometric, Wang2022Revisiting}) derive an operator equality which is equivalent to saturation of DPI for fidelity.
	This condition, for full-rank states, can be written as
	\begin{equation} \label{Eq:LRDCondition}
		\operatorname{F}(P, Q) = \operatorname{F}(\Lambda(P), \Lambda(Q)) \iff P^{-1} \# Q = \Lambda^\dagger (\Lambda(P)^{-1} \# \Lambda(Q)). 
	\end{equation}
	Indeed, by the symmetry of fidelity, one may swap $P$ and $Q$ in the above operator equality.  
	
	\section{Closed-forms for Fidelity/Bures Projections} \label{Sec:ClosedFormProjection}
	We now discuss the central results of this article, namely, closed-forms for Bures (equivalently, fidelity) projections. 
	We begin discussing the uniqueness of the Bures projection on the support of the point being projected. 
	
	\subsection{Restricted Uniqueness of Bures Projection}
	Let $(\Lambda, P, C)$ constitute the projection problem under consideration (see \eqref{Eq:ProblemStatement} for a precise definition), and $\Pi_{\Lambda, C}[P]$ denote the set optimal feasible points---projections.  
	
	We first discuss a case where there is no unique projection. 
	If $\operatorname{F}(\Lambda(P), C) = 0$, then for any $Q \in \Lambda^{-1}[C]$, DPI dictates that $\operatorname{F}(P, Q) = 0$.
	Hence every feasible point is equidistant from $P$, and thus may serve as a projection. 
	We now show that the projection is unique on $\mrm{supp}(P)$, whence it follows that the projection is unique when $P$ is invertible.
	\begin{tbox}
		\begin{theorem}	\label{Thm:ProjectionUniquenessCombined}
			Let $(\Lambda, P, C)$ constitute a projection problem and let $Q_0, Q_1 \in \plc[P]$ be arbitrary projections. 
			It holds that
			\begin{equation}
				P^0 Q_0 P^0 = P^0 Q_1 P^0,
			\end{equation}
			where $H^0$ denotes the orthogonal projector onto the support of $H$ for any Hermitian $H$. 
			Moreover, if $P$ is full-rank, then the projection is unique. 
		\end{theorem}
	\end{tbox}
	\begin{proof}
		See Appendix~\ref{App:ProjectionUniquenessCombined}.
	\end{proof}
	
	For the rest of the article, we will assume $P$ to be invertible, and thus the projection is unique.
	The data processing inequality allows us to define a simple upper-bound on the fidelity over the feasible set: for all $Q \in \Lambda^{-1}[C]$, we have 
	\begin{equation}
		\operatorname{F}(P, Q) \leq \operatorname{F}(\Lambda(P), \Lambda(Q)) = \operatorname{F}(\Lambda(P), C). 
	\end{equation}
	If there exists a feasible point $Q \in \Lambda^{-1}[C]$ saturating the above inequality, which would then be necessarily optimal, we say that \textit{the projection problem $(\Lambda, P, C)$ saturates DPI}. 
	Under suitable support conditions, we show that for a DPI-saturating projection problem, the Gamma map provides the analytic solution by leveraging the equality condition for fidelity-DPI.
	This analytic form exists for any projection problem, but whether it yields the projection depends (primarily) on the channel and, interestingly, on the saturation of the DPI of fidelity.
	
	We begin with a useful lemma which states that if a triple $(\Lambda, P, Q)$ saturates DPI for fidelity, then $P$ is a nearest point to $Q$ over all points $P'$ with $\Lambda(P') = \Lambda(P)$ and conversely $Q$ is a nearest point to $P$ over all points $Q'$ with $\Lambda(Q') = \Lambda(Q)$.  
	\begin{tbox}
		\begin{lemma} \label{Lem:DPIimpliesProj}
			Let $P, Q \in \mathbb P_\mathcal H$ and $\Lambda \in \mathrm{CPT}(\mathcal H, \mathcal K)$.
			If the triple $(\Lambda, P, Q) $ saturates DPI for fidelity, then $Q$ is a projection of $P$ to the set $\Lambda^{-1} [\Lambda(Q)]$ and $P$ is a projection of $Q$ to the set $\Lambda ^{-1} [\Lambda(P)]$:
			\begin{equation}	\label{Eq:DPIImpProjWeak}
				Q \in \Pi_{\Lambda, \Lambda(Q)}[P] \quad \text{and} \quad P \in \Pi_{\Lambda, \Lambda(P)}[Q].
			\end{equation}
			If $P$ and $Q$ are invertible, then they are the unique projections to the respective sets:
			\begin{equation}	\label{Eq:DPIImpProj}
				Q = \Pi_{\Lambda, \Lambda(Q)}[P] \quad \text{and} \quad P = \Pi_{\Lambda, \Lambda(P)}[Q].
			\end{equation}
		\end{lemma}
	\end{tbox}
	\begin{proof}
		We first show that if $(\Lambda, P, Q)$ saturates DPI, then $Q \in \Pi_{\Lambda, \Lambda(Q)}[P]$. 
		Suppose $(\Lambda, P, Q) $ saturates DPI: $\operatorname{F}(P, Q) = \operatorname{F}(\Lambda(P), \Lambda(Q))$.
		We want to show that for any $Q' \in \Lambda ^{-1}[\Lambda(Q)]$, it holds that $\operatorname{F}(P, Q') \leq \operatorname{F}(P, Q)$. 
		Towards contradiction, assume there exists $Q' \in \Lambda^{-1}[\Lambda(Q)]$ such that $ \operatorname{F}(P, Q') > \operatorname{F}(P, Q)$.
		Then we have
		\begin{equation}
			\operatorname{F}(P, Q') > \operatorname{F}(P, Q) = \operatorname{F}(\Lambda(P), \Lambda(Q)) = \operatorname{F}(\Lambda(P), \Lambda(Q')),  
		\end{equation}
		where the last equality follows from the fact that $\Lambda(Q') = \Lambda(Q)$.
		However, this violates DPI for fidelity for the triple $(\Lambda, P, Q')$. 
		Hence there cannot exist a strictly better feasible point $Q'$. 
		Thus $\operatorname{F}(P, Q') \leq \operatorname{F}(P, Q)$ for all $Q' \in \Lambda^{-1}[\Lambda(Q)]$, which implies $Q \in \Pi_{\Lambda, \Lambda(Q)}[P] $. 
		Swap $P$ and $Q$ in the above argument to obtain $P \in \Pi_{\Lambda, \Lambda(P)}[Q] $, which proves the first part.
		
		For the uniqueness, assume $P, Q > 0$.
		It follows from strict concavity of fidelity that any full-rank maximizer must necessarily be unique. 
		This concludes the proof.
	\end{proof}
	
	Next, we solve (certain) Bures projection problems analytically.
	
	\subsection{Unified closed-form for Bures Projections}
	We now define the \textit{Gamma map} which, for certain Bures projection problems, yields the solution.
	After the definition, we study the conditions under which the Gamma map yields the projection.
	\begin{dbox}	
		\begin{definition}[Gamma map] \label{Def:GammaMap} 
			Let $(\Lambda, C)$ denote a constraint pair for $C \in \mathbb P^+_\mathcal K$ and $\Lambda \in \mathrm{CPT}(\mathcal H, \mathcal K)$. 
			For any $P \in \mathbb P^+_\mathcal H$ such that $\Lambda(P) $ is full-rank, the \emph{Gamma map} $\Gamma_{\Lambda, C}: \mathbb P_\mathcal H \to \mathbb P_\mathcal H$ is defined as
			\begin{equation}	\label{Eq:GammaClosedForm}
				\begin{aligned}
					\Gamma_{\Lambda, C} [P] &:= P \star [\Lambda^\dagger (\Lambda(P) ^{-1} \# C)]^2 = [\Lambda^\dagger (\Lambda(P) ^{-1} \# C)] P [\Lambda^\dagger (\Lambda(P) ^{-1} \# C)].
				\end{aligned}
			\end{equation}
		\end{definition}
	\end{dbox}
	A sufficient condition for $\Lambda^\dagger (\Lambda(P) ^{-1} \# C)$ to be invertible is $\Lambda(P) > 0$ and $C > 0$. 
	Observe that $\Gamma_{\Lambda, C}$ map can be defined for any projection problem $(P, \Lambda, C)$.
	For certain channels---such as partial trace, pinching maps, and projective measurements with orthogonal projectors---the $\Gamma$ map is a closed-form for the projection.
	See Table~\ref{Tab:ClosedForms} for the explicit forms the $\Gamma$ map takes for these channels.
	
	We now derive two sufficient conditions for the Gamma map to yield the projection.
	\begin{tbox}
		\begin{theorem} \label{Thm:ProjectionAndFeasibility}
			Let $(\Lambda, P, C) \in \mrm{CPT}(\h, \mathcal K) \times \bbp_\h^+ \times \bbp_\mathcal{K}^+$ denote a projection problem where $\Lambda(P)$ is strictly positive. 
			Let $\{K_a\}_{a \in [n]}$ be a Kraus representation of $\Lambda$.
			The implications $1 \Longrightarrow 2 \Longleftrightarrow 3$ hold for the following statements.
			\begin{enumerate}
				\item $[K_a K_b^\dagger , \Lambda(P) ^{-1} \# C] = 0$ for all $a, b \in [n]$.
				\item $\Gamma_{\Lambda, C} [P] \in \Lambda^{-1}[C]$, or equivalently $\Lambda(\Gamma_{\Lambda, C} [P]) = C$.
				\item$ \Gamma_{\Lambda, C} [P] = \Pi_{\Lambda, C} [P]$.
			\end{enumerate}  
			Consequently, the Gamma map yields the projection if it yields a feasible point. 
		\end{theorem}
	\end{tbox} 
	\begin{proof}
		We first show $1 \Longrightarrow 2$. 
		Recall that $\Gamma_{\Lambda, C} [P] := P \star [\Lambda^{\dagger}(\Lambda(P) ^{-1} \# C)]^2$, which implies that
		\begin{equation}
			\begin{aligned}
				\Lambda(\Gamma_{\Lambda, C} [P]) &= \sum_{i=1}^{n} K_i\left[\left(\sum_{j=1}^{n} K_j^{\dagger} M K_j \right) P \left(\sum_{k=1}^{n} K_k^{\dagger}M K_k\right) \right] K_i^{\dagger} \\
				&= \sum_{i,j,k = 1}^{n} K_i K_j^{\dagger} M K_j P K_k^{\dagger} M K_k K_i^{\dagger},		
			\end{aligned}
		\end{equation}
		where $M \equiv \Lambda(P) ^{-1} \# C$. 
		Now use the assumption $[K_a K_b^{\dagger}, M] = 0$ for all $a, b \in [n]$ to obtain
		\begin{equation}
			\begin{aligned}
				\Lambda(\Gamma_{\Lambda, C} [P]) &= \sum_{i,j,k=1}^{n} M (K_i K_j^{\dagger} K_j P K_k^{\dagger} K_k K_i^{\dagger}) M \\ &=  [\Lambda(P)^{-1} \# C] \Lambda(P) [\Lambda(P)^{-1} \# C]  = C,
			\end{aligned}
		\end{equation} 
		where, in the second and third equalities, we used $\sum_{a=1}^{n} K_a^{\dagger} K_a = I_\mathcal H$ and the matrix Riccati equation respectively.
		We now show that $2 \Longrightarrow 3$. 
		Observe that
		\begin{equation}
			\Lambda(P), C > 0 \Longrightarrow \Lambda(P) ^{-1} \# C > 0\Longrightarrow \Lambda^{\dagger}(\Lambda(P) ^{-1} \# C) > 0,
		\end{equation}
		where the last implication follows from the fact that unital maps are strictly positive maps.
		We thus have $	Q \equiv \Gamma_{\Lambda, C} [P] := P \star  (\Lambda^{\dagger}(\Lambda(P) ^{-1} \# C))^2 > 0$, since we have assumed $P > 0$. 
		This allows us to write
		\begin{equation}	\label{Eq:GeoMeanProjectionProblem}
			P ^{-1} \# Q = \Lambda^{\dagger}(\Lambda(P) ^{-1} \# C). 
		\end{equation}
		Now assume $Q \equiv \Gamma_{\Lambda, C}[P] \in \Lambda^{-1}[C]$, or equivalently $\Lambda(Q) = C$.
		\eqref{Eq:GeoMeanProjectionProblem} then takes the form
		\begin{equation}
			P ^{-1} \# Q = \Lambda^{\dagger} (\Lambda(P) ^{-1} \# \Lambda(Q)),
		\end{equation}
		which, by~\eqref{Eq:LRDCondition}, is a necessary and sufficient condition for $(\Lambda, P, Q)$ to saturate DPI for fidelity.
		This further implies (by Lemma~\ref{Lem:DPIimpliesProj}) $Q \in \Pi_{\Lambda, C}[P]$. 
		Since $Q$ is full-rank, it is also the unique projection.
		
		The implication $3 \Longrightarrow 2$ follows trivially---if the $\Gamma$ map yields the projection, then, by definition, it must yield a feasible point.    
		This concludes the proof.
	\end{proof}
	The fact that the Gamma map yields the projection if it yields a feasible point provides a simple method to check if our proposed closed-form is suitable for a particular projection problem.
	
	We now look at some concrete examples of channels for which the Gamma map yields the projection. 
	These examples include the partial trace, pinching channels, and projective measurements with pairwise orthogonal projectors.

	\subsection{Partial Trace Projection}
	We now show that when the projection problem is defined by the partial trace channel, the Gamma map yields the closed-form solution. 
	We identify the relevant objects. 
	
	Let $\mathcal H = \mathcal X_1 \otimes\cdots\otimes \mathcal X_n$ and $\mathcal K = \mathcal X_k$ for some $k \in [n]$. Let $\Lambda \equiv \operatorname{Tr}_{\mathcal H \backslash \mathcal K}  \in \mathrm{CPT}(\mathcal H, \mathcal K)$ be the partial trace channel that discards every subsystem except $\mathcal K$.
	Let $C \in \mathbb P^+_{\mathcal X_k}$ and denote $\mathcal L \equiv \mathcal X_1 \otimes\cdots\otimes \mathcal X_{k-1}$ and $\mathcal R \equiv \mathcal X_{k+1} \otimes\cdots\otimes \mathcal X_{n}$ such that $\mathcal H = \mathcal L \otimes \mathcal K \otimes \mathcal R$.
	We then have the following result.
	\begin{tbox}	
		\begin{theorem}[Closed-form for marginal projection] \label{Thm:PartialTraceProj}
			Let $\mathcal H, \mathcal K$, and $ \Lambda := \operatorname{Tr}_{\mathcal H\backslash \mathcal K}$ be defined as above and $C \in \mathbb P_{\mathcal K}^+$ be arbitrary.
			For any $P \in \mathbb P_\mathcal H^+$, the Gamma map takes the form
			\begin{equation} \label{Eq:MarginalProjClosedForm}
				\begin{aligned}
					\Gamma_{\operatorname{Tr}_{\mathcal H \backslash \mathcal K}, C} [P] &= P \star (I_{\mathcal L} \otimes P_\mathcal K ^{-1} \# C  \otimes I_{\mathcal R})^2 \\
					&= P \st{K} (P_\mathcal K ^{-1} \# C)^2 \quad \quad 
					\in \operatorname{Tr}_{\mathcal H \backslash \mathcal K}^{-1} [C].
				\end{aligned} 
			\end{equation}
			where $P_\mathcal K = P_{\mathcal X_k} \equiv \mathrm{Tr}_{\mathcal H \backslash \mathcal K}[P]$ is the $\mathcal X_k$-marginal of $P$. 
			Then,
			\begin{equation}
				\Pi_{\mrm{Tr}_{\mathcal H \backslash \mathcal K}, C}[P] = \Gamma_{\mrm{Tr}_{\mathcal H \backslash \mathcal K}, C}[P] = P \st{K} (P_\mathcal K ^{-1} \# C)^2.
			\end{equation}
		\end{theorem}
	\end{tbox}
	\begin{proof}
		The form of the Gamma map follows directly from substitution (in Definition~\ref{Def:GammaMap}) as $\mrm{Tr}_{\mathcal H \backslash \mathcal K}^\dagger (K) = I_\mathcal{J} \otimes K \otimes I_\mathcal{L}$ for any $K \in \mathbb{M}_\mathcal{K}$.
		To show that the Gamma map yields the projection, we show that it yields a feasible point, which is done by computing the marginal:
		\begin{equation}
			\mrm{Tr}_{\cl{H \backslash K}}[\Gamma_{\operatorname{Tr}_{\mathcal H \backslash \mathcal K}, C} [P]] = [P \st{K} (P_\mathcal K ^{-1} \# C)^2]_{\cl{K}} = P_\cl{K} \star (P_\mathcal K ^{-1} \# C)^2 = C,
		\end{equation}
		where the second equality is by \eqref{Eq:StarProdMarginalMulti} and the last equality follows from Proposition~\ref{Prop:MatGeoMean}.
		By Theorem~\ref{Thm:ProjectionAndFeasibility} implies that it is the required projection.
	\end{proof}
	We thus have a closed-form for the projection for the partial trace channel (marginal projection).
	Set $\mathcal K = \mathbb C$, $\Lambda = \operatorname{Tr}$, and $c \in \mathbb R_+$. We have the Bures projection for the trace map to be:
	\begin{equation} \label{Eq:TrNormal}
		\Pi_{\operatorname{Tr}, c}[P] = c\frac{P}{\operatorname{Tr}[P]}.
	\end{equation}
	On choosing $c = 1$, we get the Bures projection to the set of density matrices, which is given by trace-normalization: $\Pi_{\operatorname{Tr}, 1}[P] = \frac{P}{\operatorname{Tr}[P]}$.
	
	The fact that the nearest density matrix (in fidelity/Bures distance) to a given PSD matrix is obtained by trace-normalization might not surprise many, as trace-normalization is perhaps the most \textit{natural} way to obtain a density matrix from a PSD matrix.
	However, trace-normalization need not yield the projection to the set of density matrices with respect to other distances.
	In particular, the Frobenius (Hilbert--Schmidt) projection of an arbitrary PSD matrix to the set of density matrices is \textit{not} given by trace-normalization.
	A geometric proof of this claim is seen by noting that trace-normalization corresponds to a scaling (along the line connecting the point to the origin) whereas Frobenius projection corresponds to dropping the perpendicular to the constraint-hyperplane. 
	In the case of trace distance, trace-normalization provides a \textit{non-unique} projection.
	
	\begin{tbox}
		\begin{cor} [Marginal projection for bipartite systems] \label{Cor:MarginalProjectionBipartite}
			Let $\mathcal H = \mathcal X \otimes \mathcal Y$ and $C \in \mathbb P^+_\mathcal X$.
			For any $P \in \mathbb P^+_{\mathcal X \otimes \mathcal Y }$, the $\mathcal X$-marginal projections are given by
			\begin{equation}
				\Pi_{\mrm{Tr}_\y, C}[P] = P \st{X} (P_\mathcal X ^{-1} \# C)^2.
			\end{equation}
			In particular, choosing $C = I_\mathcal X$ yields
			\begin{equation}
				\Pi_{\mrm{Tr}_\y, C}[P] = P \st{X} P_\mathcal X ^{-1}.
			\end{equation}
			where we use the fact that $A \# I = \sqrt{A}$ for any $A \geq 0$.
		\end{cor}
	\end{tbox}
	\begin{proof}
		Proof follows from relevant substitutions in Theorem~\ref{Thm:PartialTraceProj}.
	\end{proof}
	
	The operation, when $C = \ix$ has appeared previously in the literature as a \textit{partial normalization}~\cite{audenaert2008random, bruzda2009random, kukulski2021generating, Nechita2025Random} to obtain, from an arbitrary bipartite PSD matrix, a PSD matrix with one of the marginals equal to identity.
	{
		Under the Choi isomorphism, this corresponds to \textit{normalizing} a CP map to a CPT map (See Section~\ref{Sec:CPDecomp} for further details). 
		It has also appeared in the context of quantum process tomography~\cite{Xiao2022Twostage, Xiao2023AAPT, Barbera-Rodriguez2025Boosting}, where it served as a non-optimal, but closed-form, solution to optimization problems over Choi matrices of quantum channels.}
	However, to the best of our knowledge, its \textit{optimality}, in terms of being a Bures/fidelity projection, has not been previously identified.
	
	Suppose one is now interested in generalizing the \textit{partial-normalization} to a \textit{partial-scaling} to find, from a bipartite state $P \in \pxy$, a state $Q \in \pxy$ such that $Q_\x = R$.
	One might do this as 
	\begin{equation}
		P \mapsto Q \equiv P \st{X} P_\x ^{-1}  \st{X} R =  \left[\left(R \hf  P_\x \ihf \right) \ot I_\y\right]P\left[\left(R \hf  P_\x \ihf \right) \ot I_\y\right],
	\end{equation}  
	which indeed is a PSD matrix such that $Q_{\x} = R$. 
	However, this is \textit{not} optimal in the same way as the previous case as it is not the Bures projection onto the set of PSD matrices with $\x$-marginal being $\rho$. 
	The optimal operation is given by Corollary~\ref{Cor:MarginalProjectionBipartite}, which involves the matrix geometric mean.

	\subsection{Pinching Channel Projection}
	Next, we discuss the projection closed-form for \textit{pinching channels}~\cite{Watrous2018Theory, tomamichel2015quantum} which also includes the completely dephasing channel.
	Let $\mathcal E := \{E_i\}_{i \in [n]}$ be an orthogonal resolution of identity: $\sum_{i=1}^{n} E_i = I$ and $E_i E_j = \delta_{ij} E_i$.
	The pinching channel with respect to $\mathcal E$ is given by
	\begin{equation}
		\Lambda_\mathcal E(P) := \sum_{i=1}^{n} E_i P E_i.
	\end{equation}
	Observe that $\Lambda_\mathcal E$ is a self-adjoint map: $\Lambda_\mathcal E ^\dagger  = \Lambda_\mathcal E$.
	We now prove that the Gamma map yields a feasible point and thereby the projection.
	\begin{tbox}
		\begin{theorem}	\label{Thm:PinchingProjection}
			Let $\mathcal E = \{E_i\}_{i \in [n]} \subset \mathbb P_\mathcal H$ be an orthogonal resolution of $I_\mathcal H$, $\Lambda_\mathcal E$ be the corresponding pinching channel, and let $C \in \mathbb P_\mathcal H$ be in the image of $ \Lambda_\mathcal E$. 
			We then have
			\begin{equation}
				\begin{aligned}
					\Gamma_{\Lambda_\mathcal E, C}[P] = P \star (\Lambda_\mathcal E(\Lambda_\mathcal E(P) ^{-1} \# C))^2  \qquad \in \Lambda_\mathcal E ^{-1}[C]. 
				\end{aligned}
			\end{equation}
			If we pick a basis where the pinching channel visually has the effect of block-diagonalizing, then the above equation becomes
			\begin{equation}
				\Gamma_{\Lambda_\mathcal E, C}[P] = P \star \left(\bigoplus_{i=1}^n [P]_i ^{-1} \# [
				C]_i\right)^2
			\end{equation}
			where, $[P]_i$ and $[C]_i$ denotes the non-zero blocks of $E_i P E_i$ and $E_i C E_i$.
		\end{theorem}
	\end{tbox}
	\begin{proof}
		Observe that $P > 0$ implies $\Lambda_\mathcal E(P) > 0$.
		This follows from the facts that pinching channels are unital~\cite{tomamichel2015quantum, Watrous2018Theory}, and thus maps positive definite matrices to positive definite matrices (unital maps are strictly positive maps~\cite[Section 2.2]{BhatiaPD}).
		We thus have
		\begin{equation}
			(\Lambda_\mathcal E(P))^{-1}  = \left(\sum_{i=1}^{n} E_i P E_i\right)^{-1} = \sum_{i=1}^{n} (E_i P E_i)^{-1} = \sum_{i=1}^{n} (P_i)^{-1} ,
		\end{equation}
		where we denote $P_i \equiv E_i P E_i$ and the inverse is taken on the support. 
		To see this is indeed the inverse, observe that
		\begin{equation}
			\left(\sum_{i=1}^{n} E_i P E_i\right) \cdot \left(\sum_{i=j}^{n} (E_j P E_j)^{-1}\right) = \sum_{i=1}^{n} E_i P E_i  \cdot (E_i P E_i) ^{-1}  = \sum_{i=1}^{n} E_i = \mathbb I.
		\end{equation}
		Noting that $C$ is a feasible point, and thus $C = \Lambda_\mathcal E(C)$, compute
		\begin{equation}
			\Lambda_\mathcal E(P)^{-1} \# C =  \Lambda_\mathcal E(P)^{-1} \# \Lambda_\mathcal E(C) = \left(\sum_{i=1}^{n} P_i ^{-1}\right) \# \sum_{j=1}^{n} C_j = \sum_{i=1}^{n} P_i ^{-1} \# C_i. 
		\end{equation}
		where $C_i \equiv E_i C E_i$ and the last equality comes from orthogonality of the projectors.
		We thus have
		\begin{equation}
			\Gamma_{\Lambda, C} [P] =  \left(\sum_{i=1}^{n} P_i ^{-1} \# C_i\right) P \left(\sum_{j=1}^{n} P_j ^{-1} \# C_j\right) = \sum_{i,j=1}^{n} P_i ^{-1}  \# C_i P_{i j} P_j^{-1} \# C_j, 
		\end{equation}
		where we denote $P_{i j } = E_i P E_j$.
		To see this is a feasible point, let us apply the pinching channel:
		\begin{equation}
			\begin{aligned}
				\Lambda_\mathcal E(\Gamma_{\Lambda, C} [P]) &= \sum_{i=k}^{n} E_k\left( \sum_{i,j=1}^{n} P_i ^{-1}  \# C_j P_{i j} P_j^{-1} \# C_j \right) E_k \\ 
				&= \sum_{k=1}^{n} P_k ^{-1}  \# C_k P_k P_k^{-1} \# C_k = \sum_{i=1}^{n} C_k = C,
			\end{aligned} 
		\end{equation}
		as required.
		This concludes the proof.		
	\end{proof}
	
	By choosing $ \mathcal E = \{\kb{i}\}_{i \in [d]}$ to be the rank-one projectors corresponding to the computational basis vectors, one gets a closed-form for projection with respect to the completely dephasing channels.
	\begin{tbox}
		\begin{cor} [Projection with respect to completely dephasing map]
			Let $\mathcal H = \mathbb{C}^d $ and  $\Delta: \mathbb P_\mathcal H \to \mathbb P_\mathcal H$ be the completely dephasing map.
			For a diagonal positive definite matrix $C := \sum_{i=1}^{d} C_i \ketbra{i}{i} \in \mathbb P^+_\mathcal H$, the projection is given by
			\begin{equation}
				\Pi_{\Delta, C} [P] = \Gamma_{\Delta, C} [P] = \sum_{i,j=1}^{d} \sqrt{\frac{C_i C_j}{P_{ii} P_{jj}}} P_{ij} |i \rangle \langle j| = C + \sum_{i \neq j} \sqrt{\frac{C_i C_j}{P_i P_j}} P_{ij} |i \rangle \langle j|  \quad \quad \in \Delta ^{-1}[C],
			\end{equation}
			for any $P \in \mathbb P^+_\mathcal H$.
		\end{cor}
	\end{tbox}
	\begin{proof}
		Choose $E_i = \kb{i}$ for $i \in [d]$ in Theorem~\ref{Thm:PinchingProjection}.
	\end{proof}
	In particular, this corresponds to the Bures projection onto the set of PSD matrices with a specified diagonal.
	This closed-form expression might be of interest in the resource theory of coherence, although there one is interested in projecting to the set of diagonal PSD matrices~\cite{Streltsov2017Colloquium, Baumgratz2014Quantifying}, and not to the set of PSD matrices with a specific diagonal.

	\subsection{Projection With Respect to Projective Measurement}
	\label{Sec:ProjectiveMeasurementClosedForm}
	We now consider projective measurements (viewed as channels~\cite[Section 2.3]{Watrous2018Theory}) with mutually orthogonal projectors.
	Let $\mathcal E = \{E_i\}_{i \in [n]}$ be a collection of mutually orthogonal projectors that sum to identity.
	The associated measurement channel $\mathrm{M}_\mathcal E : \mathbb H_d \to \mathbb R^n$ and its adjoint $\mathrm{M}_\mathcal E^\dagger  :  \mathbb R^n \to \mathbb H_d$ are given by
	\begin{equation}
		\mathrm{M}_\mathcal E(H) := \sum_{i=1}^{n} \langle E_i, H \rangle \kb{i} \quad \text{and} \quad 	\mathrm{M}_\mathcal E^\dagger (v) = \sum_{i=1}^{n} v_i E_i,
	\end{equation}
	for $H \in \mathbb H_d$ and $v \in \mathbb R^n$.
	We now show that for projection problems defined by such measurements the Gamma map yields a feasible state and thus also the closed-form for Bures projection.
	\begin{tbox}
		\begin{theorem} \label{Thm:ProjectiveMeasurement}
			Let $\mathrm{M}_\mathcal E$ be the measurement channel associated with the projectors $\mathcal E := \{E_i\}_{i \in [n]} \subset \mathbb P_\mathcal H$ and $c \in \mathbb R_+^n$ be a positive vector. 
			Then
			\begin{equation}
				\Gamma_{\mathrm{M}_{\mathcal E}, c}[P] = P \star \left(\sum_{i=1}^{n} \sqrt{\frac{c_i}{\langle E_i, P \rangle}} E_i \right)^2 \qquad \in \mathrm{M}_\mathcal E ^{-1} [c]
			\end{equation}
			for any $P \in \mathbb P_\mathcal H^+.$	
		\end{theorem}
	\end{tbox}
	\begin{proof}
		We first show that $\Gamma_{\mathrm{M}_{\mathcal E}, c}[P] = P \star \left(\sum_{i=1}^{n} \sqrt{\frac{c_i}{\langle E_i, P \rangle} } E_i \right)^2$ by unpacking the definition. 
		To this end, observe that
		\begin{equation}
			\mathrm{M}_\mathcal E (P) ^{-1} = \sum_{i=1}^{n} \langle P, E_i \rangle ^{-1} \ketbra{i} \quad \Longrightarrow 	\quad 
			\mathrm{M}_\mathcal E (P) ^{-1} \# c = \sum_{i=1}^{n} \sqrt{\frac{c_i}{\langle E_i, P \rangle }} \ketbra{i}.
		\end{equation}
		The image under the adjoint is then obtained as $\mathrm{M}_\mathcal E^{\dagger}\left(\mathrm{M}_\mathcal E (P) ^{-1} \# c\right) = \sum_{i=1}^{n} \sqrt{\frac{c_i}{\langle E_i, P \rangle }} E_i,$
		which leads to the form
		\begin{equation}
			\Gamma_{\mathrm{M}_{\mathcal E}, c}[P] = P \star \left(\sum_{i=1}^{n} \sqrt{\frac{c_i}{\langle E_i, P \rangle}} E_i \right)^2
		\end{equation}
		as claimed.
		To show that $\Gamma_{\mathrm{M}_{\mathcal E}, c}[P]$ is a feasible point, compute
		\begin{equation}
			\mathrm{M}_\mathcal E \left( \Gamma_{\mathrm{M}_{\mathcal E}, c}[P]\right) = \sum_{i=1}^{n} \left\langle E_i, \Gamma_{\mathrm{M}_{\mathcal E}, c}[P] \right \rangle \ketbra{i} = \sum_{i=1}^{n} c_i \ketbra{i} = c.
		\end{equation}
		The second equality follows from
		\begin{equation}
			\left\langle E_i, \Gamma_{\mathrm{M}_{\mathcal E}, c}[P] \right \rangle = \sum_{j,k=1}^{n} \sqrt{\frac{c_j c_k}{\langle E_j, P \rangle \langle E_k, P \rangle }} \operatorname{Tr}\left[E_i E_j P E_k \right] = \frac{c_i}{\langle E_i, P \rangle} \langle E_i, P \rangle = c_i,
		\end{equation}
		where we have used $E_i E_j = \delta_{i j}E_i$ as $\mathcal E := \{E_i\}_{i \in [n] }$ form a collection of mutually orthogonal projectors.
		We have thus shown $\Gamma_{\Lambda_\mathcal E, C}[P] \in \mathrm{M}_\mathcal E ^{-1}[c]$ which, by Theorem~\ref{Thm:ProjectionAndFeasibility}, implies it is also the projection $\Pi_{\mathrm{M}_\mathcal E, c}[P]$.	
		This concludes the proof.
	\end{proof}
	
	\subsection{Ensemble Projection} \label{Sec:EnsProj}
	The final type of projection we discuss is the projection of \textit{ensembles} of PSD matrices.
	Though this is a special case of the partial-trace projection, it is useful to study it separately.
	We first define the Bures distance between two ensembles.
	\begin{dbox}
		\begin{definition} \label{Def:BuresDistEnsemble}
			Let $\mathcal{P} := (P_i)_\inn \subset \bbp_d$ and $\mathcal{Q} := (Q_i)_\inn \subset \bbp_d$ be two PSD ensembles of equal length. 
			The squared Bures distance between $\mathcal{P}$ and $\mathcal{Q}$ is defined as
			\begin{equation} \label{Eq:BuresDistEnsemble}
				\mrm{B}(\cl{P}, \cl Q) := \sumin \mrm{B}(P_i, Q_i). 
			\end{equation}
		\end{definition}
	\end{dbox} 
	The equality follows from the fact that the squared Bures distance (\textit{and square-root} fidelity) is additive under direct sum~\cite{Watrous2018Theory, Wilde2013}.
	That is,
	\begin{equation}
		\mrm{B}\left( \sumin \kb{i} \ot P_i, \sumin \kb{i} \ot Q_i\right) = \sumin \mrm{B}(P_i, Q_i),
	\end{equation}
	for any pair of tuples of PSD matrices $(P_i)_{\inn}, (Q_i)_{\inn} \subset \mathbb{P}_d$. 
	The validity of Definition~\ref{Def:BuresDistEnsemble} as a distance follows from the fact that the (PSD) block-matrices of the above form are bijective to the set of $n$-length ensembles of $d \times d$ PSD matrices, and Bures distance is a bona fide distance over PSD matrices.  
	
	\paragraph{Bures projection of ensembles.} The notion of Bures distance between ensembles allows one to study the \textit{Bures projection} of ensembles.
	Let $\mathcal P = (P_1, \ldots, P_n)$ be an ensemble of positive definite matrices in $\mathbb P_d^+$ with $ P := \sum_{i=1}^{n} P_i$. Let $Q \in \mathbb P^+_d$ and define
	\begin{equation}
		\mathrm{Dec}_n (Q) := \left\{\mathcal Q:= (Q_1, \ldots, Q_n) \in {\mathbb P_d}^{\times n} :\sum_{i=1}^{n} Q_i = Q\right\}
	\end{equation}
	to be the set of all $n$-length (PSD) \textit{decompositions} of $Q$ (which is a spectrahedron~\cite{chiribella2024extreme}). 
	The problem of interest is to project $\mathcal P$ to this compact and convex set.
	Indeed we have  
	\begin{equation} \label{Eq:EnsembleProjEquivalence}
		\Pi_{\mrm{Dec}_n(Q)}[\cl P] :=  \argmin_{\mathcal{Q} \in \mrm{Dec}_n(Q)}\mrm{B}(\cl P, \cl Q) = \argmin_{(Q_i)_i \in \mathrm{Dec}_n (Q)}  \sum_{i=1}^{n} \mathrm{B}(P_i, Q_i) = 	\argmax_{(Q_i)_i \in \mathrm{Dec}_n (Q)} \sum_{i=1}^{n} \mathrm{F}(P_i, Q_i).
	\end{equation}
	{The last equality follows from the fact that the trace terms ($\sumin \tr{P_i+Q_i}$) are constant for a given problem.}  
	Note that choosing $Q = I_d$ gives the projection to the set of $n$-outcome POVMs.
	
	The following lemma allows us to transform the ensemble projection problem to an instance of marginal projection and allow us to use Theorem~\ref{Cor:MarginalProjectionBipartite} to derive a closed-form expression for the optimal ensemble elements. 
	
	\begin{tbox}
		\begin{lemma} \label{Lem:DPIBlockDiagonal}
			Let $\cl{P} := (P_i)_\inn  \subset \bbp_\y^+$, $\x \equiv \cn$, and $\mathbf{P} := \sumin \kb{i} \ot P_i\in \mathbb{P}_{\mathcal X \otimes \mathcal Y}$.
			The Bures projection of $\mathbf{P}$ onto the feasible set
			\begin{equation}
				\mathrm{Tr}_\mathcal X^{-1}[Q]
				:= \left\{ \mathbf{Q} \in \mathbb{P}_{\mathcal X \otimes \mathcal Y} \,:\, \mathrm{Tr}_\mathcal X[\mathbf{Q}] = Q \right\}
			\end{equation}
			is a block-diagonal matrix with the same block structure as $\mathbf{P}$.
		\end{lemma}
	\end{tbox}
	\begin{proof}
		This follows from the data-processing inequality for the fidelity (equivalently Bures distance).
		Let $\mathbf{Q} := \sumin \ketbra{i}{j} \ot Q_{ij}$ be the block-matrix decomposition of arbitrary $\mathbf{Q} \in \mathrm{Tr}_\mathcal X^{-1}[Q]$. 
		The DPI for fidelity, with respect to the pinching channel $\Lambda_\mathcal{E} \in \mrm{CPT}(\xoy, \xoy)$ with Kraus operators $\{\kb{i} \ot \iy\}_\inn$ then implies
		\begin{equation}
			\operatorname{F}\left( {\sum_{i=1}^{n} \ketbra{i} \otimes  P_i}, \sum_{i,j=1}^{n} |i \rangle \langle j| \otimes Q_{i j}\right)  \leq \operatorname{F}\left( \sum_{i=1}^{n} \ketbra{i} \otimes  P_i, \sum_{i=1}^{n} |i \rangle \langle i| \otimes Q_{i i}\right).
		\end{equation}
		Since $ \sumin \ketbra i j \ot Q_{ij} \in \mathrm{Tr}_\mathcal X^{-1}[Q]\Rightarrow \sumin \kb i \ot Q_{ii} \in \mathrm{Tr}_\mathcal X^{-1}[Q]$, we have that the projection of any block-diagonal matrix will also be block-diagonal. 
		This concludes the proof.
	\end{proof}
	
	Using a direct-sum construction and the closed-form for marginal projection, we now derive a closed-form expression for the optimal ensemble.	
	\begin{tbox}
		\begin{theorem} [Ensemble projection] \label{Thm:EnsembleProjection}
			Let $\mathcal P = (P_1, \ldots, P_n) \subset \mathbb P^+_\mathcal Y$ be a PSD ensemble and let $Q \in \mathbb P^+_\mathcal Y$.
			Then the ensemble projection $ \Pi_{\mrm{Dec}_n(Q)}[\cl P] \equiv \mathcal Q = (Q_1, \ldots, Q_n) \subset \mathbb P^+_\mathcal Y$ is given by
			\begin{equation} \label{Eq:EnsembleProjCF}
				Q_i = P_i \star (P^{-1} \# Q)^2 = (P^{-1} \# Q) P_i (P^{-1} \# Q) 
			\end{equation}
			for all $\inn$, where $P := \sum_{i=1}^{n} P_i$.
		\end{theorem}
	\end{tbox}
	\begin{proof}
		Let $\x \equiv \mathbb{C}^d$.
		The proof applies the closed-form for marginal projection to the direct-sum matrix $\mathbf{P}:= \sumin \kb{i} \ot P_i \in \mathbb{P}_{\xoy}$ to project it to the set $\mathrm{Tr}_\mathcal X ^{-1} [Q] := \{\mathbf{Q} \in \mathbb P_{\mathcal X \otimes \mathcal Y } : \mathrm{Tr}_\mathcal X[\mathbf{Q}] = Q \}. $
		Using the closed-form for marginal projection (\eqref{Eq:MarginalProjClosedForm}), we get
		\begin{equation}
			\Pi_{\mrm{Tr}_\x, Q}{[\mathbf{P}]} = \mathbf{P} \st{Y} (P \iv \# Q)^2= \sumin \kb{i} \ot P_i \star (P\iv \# Q)^2 \equiv \sumin \kb{i} \ot Q_i. 
		\end{equation}
		The optimum is achieved at a block-diagonal matrix, and we return the corresponding blocks as the optimal ensemble.
	\end{proof}
	
	Thus, we have a remarkably simple closed-form for the nearest ensemble (with a given sum) to a given ensemble of PSD matrices.
	By~\eqref{Eq:EnsembleProjEquivalence}, this ensemble also maximizes the \textit{total (square root) fidelity} with the original ensemble.
	
	The above closed form should find applications in optimization problems involving fidelity.
	For example, if we choose $\mathcal P \equiv (\rho_1, \ldots, \rho_n)$ to be a possibly weighted ensemble of states and $Q = \sigma$ to be a fixed density matrix, then the closed-form gives the closest ensemble to $\mathcal P$ which is a decomposition of $\sigma$.
	By the bijection between PSD and weighted quantum states, this optimal ensemble is equivalent to a weighted ensemble of quantum states.
	
	\textbf{Remark.} We note an interesting observation regarding the form of ensemble projection.
	The ensemble elements of $\Pi_{\mrm{Dec}_n(Q)}[\mathcal{P}]$ is given by $Q_i = P_i \star (P \iv \# Q)^2 \equiv \frac{P_i}{(Q ^{-1} \# P)^2}$ for all $\inn$. 
	Suppose all the matrices involved pairwise commute.
	Then the relation reduces to
	\begin{equation}
		Q_i  = \frac{Q}{P} P_i,
	\end{equation}
	which is reminiscent of the \textit{scaling} one does in the case of vectors, where to obtain a vector that sums to $\bar q \in \mathbb R_+$ from a vector $p \in \mathbb{R}_+^d $, one performs the scaling $p \mapsto \frac{\bar q}{\bar p} p$ for $\bar p := \sum_{i=1}^{d} p_i$.
	Thus, the Bures projection of ensembles can be thought of as a non-commutative scaling operation. 
	
	\paragraph{Ensemble projection to POVMs.} We return to the non-commutative setting to discuss a particularly interesting choice of $Q$---namely $Q = I_\mathcal Y$.
	For this choice, the ensemble is projected to the compact set of $n$-outcome measurements (POVMs). 
	The closed-form for ensemble projection then yields the familiar \textit{pretty good measurement} (PGM)~\cite{belavkin1975optimal, hausladen1994pretty} as the following corollary reveals.
	\begin{tbox}
		\begin{cor}  \label{Cor:PGMasEnsembleProjection}
			Let $\mathcal P = (P_1, \ldots, P_n) \subset \mathbb  P_d^+$ be an ensemble and $P := \sum_{i=1}^{n} P_i$.
			The ensemble projection of $\mathcal P$ onto the set of $n$-outcome POVMs (equivalent to $\mathrm{Dec}_n(I_d)$) is given by
			\begin{equation}
				Q_i = \frac{P_i}{P} = P ^{-\frac12} P_i P ^{-\frac12},
			\end{equation} 
			which is the `pretty good measurement' associated with the ensemble $\mathcal P$.
		\end{cor}
	\end{tbox}
	\begin{proof}
		The proof follows from Theorem~\ref{Thm:EnsembleProjection} on setting $Q = I_d$.
	\end{proof}
	As we will discuss in Section~\ref{Sec:PGMManifestation}, this result endows the PGM with new operational and geometric interpretations.
	Having discussed the explicit form of the Bures projection for a few important channels, we now look at a related, but slightly different concept, which we call the \textit{prior-channel decomposition} of CP maps.

	\section{Prior-Channel Decomposition of Completely Positive Maps} \label{Sec:CPDecomp}
	We now discuss how every CP map\footnote{with appropriate conditions on the support.} admits a \textit{unique} decomposition into a quantum channel and a \textit{prior} PSD matrix,\footnote{We expect some of these results to be already known, but were unable to find a reference.} with the channel and prior acting as the normalized and unnormalized components, respectively.  
	We will then discuss the geometric aspects of this decomposition and how it naturally relates to the fidelity/Bures projections at the level of Choi matrices, and, equivalently, to the Frobenius projections at the level of Stinespring operators. 
	We will use this decomposition (and its converse operation) extensively in Section~\ref{Sec:ApplicationsAndManifestations} to discuss the applications of our results. 
	
	To motivate the results of this section, we first discuss a natural decomposition of PSD matrices $\mathbb P_\mathcal X$---into density matrices (normalized component) and positive scalars (unnormalized component)---and its geometry.
	\begin{tbox}
		\begin{prop} \label{Prop:PSDDecomp}
			Let $P \in \mathbb{P}_\x$ such that $P \neq 0$.
			The following statements hold. 
			\begin{enumerate}
				\item There exists a unique pair $(\rho, s) \in \dx \times \mathbb{R}^+$ such that $P = s \rho$.
				\item $\rho = \Pi_{\dx}[P] = \frac{P}{\tr{P}}$.
				\item For any purification $|u \ra  \in \mathrm{Pur}_\z (P)$, it holds that $\Pi_{\mathbb{S}_{\z \ot \x}}^{\mathrm{Eu}}[|u \ra ] \in \mathrm{Pur}_{\z}(\rho)$,
			\end{enumerate} 
			where $\Pi^\mrm{Eu}_{\mathbb{S}_{\z \ot \x}}[\cdot]$ denotes the Euclidean projection to the unit sphere in ${\z \ot \x}.$
			Moreover, $s = \tr{P} = \|u\|^2_2.$
		\end{prop}
	\end{tbox}
	\begin{proof}
		The existence is seen by choosing $\rho \equiv \frac{P}{\tr{P}}$ and $s \equiv \tr{P}$. 
		For uniqueness, assume $s \rho = P = t \sigma$ for $s, t \in \mathbb{R}^+$ and $\rho, \sigma \in \dx$. 
		Take trace across to get $s = t$ whence it follows that $\rho = \sigma$, thereby demonstrating the uniqueness.
		The fact that $\rho \equiv \frac{P}{\tr{P}} = \Pi_{\dx}[P]$ follows from \eqref{Eq:TrNormal}. 
		
		For the third part, let $\ket u \in \z \ot \x $ be an arbitrary purification of $P$. 
		The Euclidean projection of an arbitrary vector $u$ onto the unit sphere $\mathbb{S}_{\z \ot \x} := \{u \in \z \ot \x : \la u, u \ra = 1\}$ is given by $\ell_2$-normalization: $\Pi_{\mathbb{S}_{\z \ot \x}}^{\mathrm{Eu}}[\ket u] := \frac{\ket u}{\|u\|_2} \equiv \ket v$.
		It follows that $v$ is a purification of $\rho$:
		\begin{equation}
			\trz{\kb{v} } = \frac{1}{\|u\|_2^2} \trz{\kb{u}} = \frac{1}{\tr{P}} P = \rho.
		\end{equation} 
		This concludes the proof.
	\end{proof}
	Unsurprisingly, every non-zero PSD matrix can be uniquely decomposed into a density matrix and a scalar, such that their product returns the original PSD matrix.
	Moreover, the density matrix component is Bures projection of the original PSD matrix onto the set of density matrices. 
	Observing from the \textit{purification space}, we see that the Euclidean projection, onto the unit sphere, of any purification of the original PSD matrix yields a purification of the Bures projection of the PSD matrix onto the set of density matrices.
	
	Interestingly, the above proposition can be generalized from PSD matrices to CP maps as follows.
	In the following proposition and henceforth, For any PSD matrix $P \in \mathbb{P}_{\x}$, we define the CP map $\kp_P \in \mathrm{CP}(\x, \x)$ as $\kp_P (X) :=  X \star P$, and $\kp_P ^{-1} \equiv \kp_{P^{-1}}$ with negative powers defined on the support. 
	\begin{tbox}	
		\begin{prop} \label{Prop:ChPrDecomp}
			Let $\Theta \in \cpxy$ such that $\tiy > 0$.
			The following statements hold.
			\begin{enumerate}
				\item There exists a unique pair $(R, \Phi) \in \mathbb{P}_\x^+ \times \cpt $ such that $\Theta = \Phi \circ \kp_R$.
				\item $\mfr{C}(\Phi) = \Pi_{\mrm{Tr}_\y, \ix}[\mfr{C}(\te)] = \mfr{C}(\te) \st{X} {{[\cth]}_{\x}^{\ic}}^{-1}$, or equivalently, $\Phi = \Pi_{\cpt}^{\mrm{CB}}[\Theta]$.
				\item For any Stinespring operator $A \in \mathrm{SR}_\z (\Theta)$, it holds that $\Pi_{\mathbb{U}_{\x, \y \ot \z}}^{\mathrm{F}}[A] \in \mathrm{SR}_{\z}(\Phi)$,
			\end{enumerate}
			where $\Pi^\mrm{F}_{{\mathbb U}_{\x, \y \ot \z}}$ denotes Frobenius projection to the set of isometries from $\x$ to $\y \ot \z$. 
			Moreover, $R = \tiy = \mathrm{Tr}_{\y}[\fc(\Theta)]^\ic = A^\dagger A$.
		\end{prop}
	\end{tbox}
	\noindent\textbf{Remark.} If $\Theta \in \mathrm{CP}(\x, \y)$ is such that $\tiy$ is rank-deficient, then Statement 1 must be modified such that $\Phi$ is TP on $\mathrm{supp}(\tiy)$ and the \textit{prior} $R$ is unique on $\mathrm{supp}(\tiy)$.
	Moreover, the Frobenius projection in Statement 3 would yield a partial isometry.
	
	\begin{proof}
		The existence is seen by choosing $R \equiv \tiy$ and setting $\Phi \equiv \te \circ \mathrm{K}_{R}^{-1}$. 
		To see the uniqueness, assume the existence of pairs $(R, \Phi)$ and $(S, \Psi)$ such that $\Theta$:
		\begin{equation}
			\Phi \circ \kp_{R} = \Theta = \Psi \circ \kp_{S} \quad \Longleftrightarrow \quad \kp_R \circ \Phi^\dagger  = \Theta^\dagger  = \kp_S \circ \Psi^\dagger . 
		\end{equation}
		By acting the adjoint maps on the RHS on $\iy$, we get $R = \tiy = S$ as $\Psi^\dagger , \Phi^\dagger $ are unital maps.
		Moreover, the positive-definiteness of $R \equiv \tiy$ implies that $\Phi = \Theta \circ \kp_{R}^{-1} = \Psi$. 
		The second part directly follows from Corollary~\ref{Cor:MarginalProjectionBipartite}.
		
		We now prove the third part. 
		Let $A \in \mm_{\x, \y \ot \z}$ be an arbitrary Stinespring representation of $\te$. 
		The Frobenius projection to the set of isometries of any matrix $K$ has a closed-form solution given by the polar decomposition.
		That is, if $K = U |K|$ is the polar decomposition of $K$ where $|K| := \sqrt{K^\dagger K}$, then $U = K|K| ^{-1} $ is the nearest isometry, in Frobenius distance, to $K$. 
		If $|K|$ is rank-deficient, then we take $|K|^{-1} $ to be its pseudoinverse, which leads to $U$ being a partial isometry.
		
		Applying this to the Stinespring operator $A$, we have $V \equiv A|A| ^{-1}$ to be the Frobenius projection of $A$ to the set of isometries. 
		Observe that $A^\dagger A = \tiy = R$, which implies that projection is $ V = A |A| ^{-1} = A R \ihf$. 
		Indeed, it follows that $V$ is a Stinespring representation of $\Phi = \te \circ \kp_{R} ^{-1} $, concluding the proof. 
	\end{proof}
	
	We now discuss the implications of this proposition.
	We see that CP maps admit a very natural decomposition into a prior state and a CPT map, generalizing how a PSD matrix can be decomposed into a density matrix and a positive scalar.
	In fact, by viewing PSD matrices in $\px$ as linear maps from $\mathbb{C}$ to $\x$ and applying~Proposition~\ref{Prop:ChPrDecomp}, one recovers Proposition~\ref{Prop:PSDDecomp}.   
	We call this decomposition the \textit{prior-channel} decomposition of CP maps.
	\begin{dbox}
		\begin{definition}[Prior-channel decomposition] Let $\Theta \in \cpxy$ be an arbitrary CP map. 
			Let 
			\begin{equation}
				R \equiv \tiy \qaq \Phi \equiv \Theta \circ \kp_R \iv.
			\end{equation}
			The pair $(R, \Phi)$ constitute the (unique) \emph{prior-channel decomposition} of $\Theta$.
		\end{definition}
	\end{dbox}

	This decomposition can be seen as an extension of Choi--Jamiołkowski isomorphism.
	Recall that the Choi--Jamiołkowski isomorphism bijectively associates quantum channels to not all of bipartite PSD matrices, but only a subset of them (those with the identity as the input-space marginal). 
	The above result states that \textit{every} bipartite PSD matrix can be uniquely\footnote{Under mild conditions on the support.} associated with a \textit{pair} of channel and a prior (unnormalized) state.
	This coincides with the standard Choi--Jamiołkowski isomorphism if the prior state is the identity matrix. 
	
	Similar to how the Bures projection of $P \in \px$ onto density matrices gave the normalized component $\rho  \equiv \frac{P}{\tr{P}}$, the \textit{Choi--Bures} projection (see Definition~\ref{Def:CBProj}) of $\Theta$ to the set of CPT maps yields its channel component.
	Moreover, similar to how---for PSD matrices---the Euclidean projection of any purification (onto the unit sphere) yields a purification of Bures projection, we have that the Frobenius projection of any Stinespring operator (onto the set of isometries) yields a Stinespring operator of the Bures projection.   
	
	We make some further observations. 
	In line with the above nomenclature, let us call $\tiy$ the \textit{prior} of any $\Theta \in \cpxy$.
	The set of CP maps with a given prior forms a spectrahedron. 
	Given $R \in \mathbb P_\mathcal X$, we denote $\mathrm{CP}_{R}[\mathcal X, \mathcal Y]$ to be the spectrahedron of CP maps whose prior is $R$:
	\begin{equation} \label{Eq:CPPriorSet}
		\mathrm{CP}_R (\mathcal X, \mathcal Y) = \{\Theta \in \mathrm{CP} (\mathcal X, \mathcal Y) : \Theta^\dagger (I_\mathcal Y) = R\}. 
	\end{equation}
	At the level of Choi representation, these are bipartite PSD matrices of the form $\fc[\mathrm{CP}_R] = \{P \in \pxy : P_\x = R^\ic \}$. 
	Indeed we have that $\mathrm{CP}_{\ix} (\mathcal X, \mathcal Y) = \cpt$.   
	Akin to computing the \textit{Choi--Bures} projection of a CP map to the set of CPT maps, one can Choi--Bures-project an arbitrary CP map to the set of CP maps with an arbitrary prior as follows.
	\begin{tbox}
		\begin{prop} \label{Prop:CBProjPriorS}
			For an arbitrary $\te \in \cpxy$, let $\te \sim (R, \Phi)$. 
			Then the Choi--Bures projection of $\te$ to the set of CP maps with prior $S$ is given by
			\begin{equation}
				\Xi \equiv \Pi^{\mathrm{CB}}_{\mathrm{CP}_S(\x, \y)}[\te] = \Theta \circ \kp_{ {(R ^{-1} \# S)}^2}.
			\end{equation}
		\end{prop}
	\end{tbox}
	\begin{proof}
		One may verify that $\Xi$ as defined above is indeed an element of $\mathrm{CP}_S(\x, \y)$ as
		\begin{equation}
			\Xi^\dagger (\iy) = \kp_{ {(R ^{-1} \# S)}^2} \circ \tiy = \kp_{{(R ^{-1} \# S)}^2} (R) = S,
		\end{equation}
		where the last equality follows from the matrix Riccati equation.
		For a formal proof, we use the closed-form for marginal projection (\eqref{Eq:MarginalProjClosedForm}).
		\begin{equation}
			\begin{aligned}
				\fc(\Xi) \equiv \Pi_{\mrm{Tr}_\y, S^\ic} [\fc(\te)] 
				&= \fc(\te) \st{X} ([\cphi]_\x ^{-1} \# S^ \ic)^2  
				= \fc(\te) \st{X} ({R ^ {-1}}^\ic  \# S^\ic)^2 \\ 
				&= \fc(\Theta) \st{X} [(R^{-1} \# S)^2]^\ic  = \fc \left(\te \circ \kp_{(R^{-1} \# S)^2} \right),
			\end{aligned}
		\end{equation} 
		as claimed.
	\end{proof}
	
	Converse to decomposing a CP map to a CPT map and a prior, one can \textit{join} a CPT map and a prior, a process we call the \textit{CP extension} (of the CPT map by the prior). 
	This generalizes the basic operation of scaling a density matrix with an arbitrary positive scalar. 
	\begin{dbox}
		\begin{definition}[CP extension of a channel (by a prior)] \label{Def:CPExt}
			Let $\Phi \in \cpt$ and $R \in \px$. 
			The \emph{CP extension of $\Phi$ by $R$} is defined as $\Phi \circ \kp_R \in \cpxy$. 
		\end{definition}
	\end{dbox}
	\noindent \textbf{Remark.} Geometrically, this corresponds to the Choi--Bures projection of $\Phi$ to $\mathrm{CP}_R(\x, \y)$: $\Phi \circ \kp_R = \Pi_{\mrm{CP}_R(\x, \y)}^{\mrm{CB}}[\Phi]$, or equivalently,
	\begin{equation} \label{Eq:CPExt}
		\fc (\Phi \circ \kp_R) = \fc(\Phi) \st{X} R^\ic = \Pi_{\mrm{Tr}_\y, R^\ic} [\cphi]. 
	\end{equation}
	It is easy to see that the prior of the CP map $\Phi \circ \kp_R$ is indeed $R$, and the prior-channel decomposition of $\Phi \circ \kp_R$ returns the pair $(R, \Phi)$. 
	We will use these decompositions and extensions extensively in the subsequent sections as it both unburdens us of cumbersome notation and gives the intuition behind many of the results discussed.

	\section{Applications} \label{Sec:ApplicationsAndManifestations}
	We now discuss various applications of Bures projection.
	The \textit{zeroth} application of our results is that we give explicit instances of triples that saturate the DPI for fidelity.
	That is, for all the projection problems $(\Lambda, P, C)$ we discuss above, the triple $(\Lambda, P, \Gamma_{\Lambda, C}[P])$ saturates the DPI for fidelity, and equivalently, Sandwiched R{\'e}nyi divergence of order $\alpha = 1/2$.
	
	The first application is a rather direct implication of the closed-forms---namely, we now have simple closed-forms for projections onto various sets of interest in quantum information. 
	\subsection{Bures Projection onto Sets of Interest}
	As a direct corollary of the results from Section~\ref{Sec:ClosedFormProjection}, we now have simple closed-forms for projections with respect to Bures distance, fidelity, and purified distance\footnote{provided the objects involved are subnormalized states.} to the following sets.
	\begin{enumerate}[itemsep=-0.5ex,partopsep=1ex,parsep=1ex]
		\item A multipartite PSD matrix $P \in \mathbb{P}_{X \otimes \y \ot \z  } $ to the set of PSD matrices with one of the marginals equal to some fixed PSD matrix $Y \in \mathbb P_\y$.
		\item A CP map to the set of CPT, CPU maps, or to the set of CP maps with a given prior. 
		\item An ensemble of PSD matrices to the set of decompositions of a given matrix (including the set of POVMs).
		\item A PSD matrix to the set of PSD matrices with a given (block) diagonal, or more generally, a given output to a pinching map.
		\item A density matrix to the set of density matrices with a given measurement probabilities under a projective measurement (with mutually orthogonal projectors). 
	\end{enumerate}
	
	\subsection{Geometric Interpretations of the Pretty Good Measurement} \label{Sec:PGMManifestation}
	The \textit{Pretty good measurement}~\cite{belavkin1975optimal, hausladen1994pretty, holevo1978asymptotically}  (a.k.a. \textit{square-root measurement}) is a canonical way of associating a POVM to a (possibly weighted) ensemble of states. 
	It has found applications in quantum state discrimination~\cite{barnum2002reversing, Watrous2018Theory, bae2015quantum}, quantum learning and tomography~\cite{arunachalam2018optimal, Haah2016Sampleoptimal}, hidden subgroup problem~\cite{bacon2005optimal, Hayashi2008Quantum}, and quantum communication~\cite{Hausladen1996Classical, beigi2014quantum, cheng2023simple}, among others. 
	Recent works have extended it to the continuous variable setting as well~\cite{Mishra2024Pretty, Mishra2025Nearoptimal}.
	In many of these settings, the PGM is near-optimal and turns out to be optimal in problems with a high degree of symmetry~\cite{Eldar2002Quantum, eldar2004optimal, DallaPozza2015Optimality, iten2017pretty, leditzky2022optimality, Zhou2025Distinguishability}.
	We first look at the definition of the PGM. 
	\begin{dbox}
		\begin{definition}[Pretty good measurement]
			Let $\mathcal{P} := (P_i)_{\inn} \subset \bbp_d$ be PSD ensemble and define $P := \sumin P_i$.
			The \emph{pretty good measurement} associated with $\mathcal{P}$ is defined as
			\begin{equation}
				\mathcal{E} := \left(P \ihf P_i P \ihf \right)_\inn
			\end{equation}
		\end{definition}
	\end{dbox}
	In quantum information, the ensemble $\mathcal{P}$ is typically derived from a weighted ensemble of states $(w_i, \rho_i)_\inn$ where $(w_i)_\inn$ is a probability vector and $(\rho_i)_\inn\subset \mathbb{D}_d$ is a collection of density matrices. 
	The associated ensemble is then $\mathcal{P} = (P_i := w_i \rho_i)_\inn$. 
	
	Various explanations have been put forward for the (near) optimality of the PGM. 
	In the context of state discrimination, \textcite{hausladen1994pretty} interprets that the mutual information (for a fixed ensemble and variable measurement) has a \textit{small} gradient at PGM (at the optimal POVM, the gradient would be zero).
	\textcite{Cheng2025Error} discusses how a variant of the PGM can be seen as a randomized Holevo--Helstrom measurement. 
	
	Our results on ensemble projection (Theorem~\ref{Thm:EnsembleProjection}) yields novel geometric and operational interpretations of the pretty good measurement (PGM), formalized below.
	\begin{tbox}
		\begin{theorem}[Novel interpretations of the PGM]
			\label{Thm:PGMInterpretations}
			Let $\mathcal{P} := (P_i)_\inn \subset \bbp_d^+$ be an ensemble, $P := \sum_{i=1}^n P_i$, and let $\mathcal{E} := (E_i)_{i \in [n]} = (P^{-1/2} P_i P^{-1/2})_\inn$ be the associated PGM. Then,
			\begin{enumerate}
				\item $\mathcal{E}$ is the ensemble projection of $\mathcal{P}$ onto the set of $n$-outcome POVMs,
				\begin{equation}
					\Pi_{n\textrm{-POVM}}[\mathcal{ P}] = (P^{-1/2} P_i P^{-1/2})_\inn .
				\end{equation}
				\item $\mathcal{E}$ maximizes the total square-root fidelity over all POVMs:
				\begin{equation}
					\mathcal{E} = \argmin_{(E'_i)\textrm{ is a POVM}} \sum_{i=1}^n \mathrm{B}(P_i,E'_i)
					= \argmax_{(E'_i)\textrm{ is a POVM}} \sum_{i=1}^n \mathrm{F}(P_i,E'_i).
				\end{equation}
			\end{enumerate}
			Here $n$-POVM $\equiv \mrm{Dec}_n(I_d)$ is the set of $n$-outcome POVMs acting on $\bb C^d$.
		\end{theorem}
	\end{tbox}
	
	Thus, the PGM admits both a geometric interpretation as a Bures projection and an operational one as the POVM maximizing total (square root) fidelity with the ensemble elements.
	Theorem~\ref{Thm:PGMInterpretations} also yields geometric interpretations for the \textit{Pretty bad measurement}~\cite{McIrvin2024Quantum} (with applications in quantum state exclusion),  $\alpha$-power PGMs~\cite{tyson2009error,tyson2009two} (with applications in state discrimination for geometrically uniform states~\cite{Zhou2025Distinguishability}).
	
	An equivalent geometric interpretation for the PGM holds in terms of Choi matrices. Define
	\begin{equation}
		\hat P := \sum_{i=1}^n P_i^{\mathsf{T}} \otimes |i\rangle\!\langle i|
		\in \mathbb{P}_{\mathcal X \otimes \mathcal Y},
		\qquad \mathcal Y \cong \mathbb{C}^n .
	\end{equation}
	The associated CP map $\Xi \equiv \fc \iv (\hat P)$ has the action $   \Xi(X) = \sum_{i=1}^n \langle P_i, X \rangle \, |i\rangle\!\langle i|$ (see also~\cite[Eq.~43]{Coutts2021}).
	The Choi--Bures projection of $\Xi$ onto the set of CPT maps yields the PGM $\Psi(X) = \sum_{i=1}^n \langle P^{-1/2} P_i P^{-1/2}, X \rangle \, |i\rangle\!\langle i|.$
	The equivalence between Choi--Bures and Stinespring--Frobenius projections (Proposition~\ref{Prop:PurificationChannels}), implies that least-squares problem
	\begin{equation}
		\argmin_{(E_i)_{i=1}^n \text{ is a POVM}}
		\sum_{i=1}^n \left\| E_i \hf  - P_i \hf  \right\|_2^2 ,
	\end{equation}
	is solved at the PGM $E_i = P \ihf P_i P \ihf$.
	The pure-state, equal-weight case was previously derived in~\cite{Eldar2002Quantum}.
	
	\paragraph{Implications for quantum state discrimination.}
	Given states $(\rho_i)_\inn$ with probability weights $w = (w_1, \ldots w_n) $, the quantum state discrimination problem~\cite{Watrous2018Theory, bae2015quantum} is
	\begin{equation}\label{Eq:StateDisc}
		\begin{aligned}
			\mathrm{maximize:} \quad & \sum_{i=1}^n \langle P_i, E_i \rangle \\
			\mathrm{subject~to:} \quad & (E_1,\dots,E_n) \text{ is a POVM},
		\end{aligned}
	\end{equation}
	where $P_i := w_i \rho_i$ for $\inn$. 
	Except for $n=2$~\cite{HOLEVO1973Statistical,Helstrom1969Quantum,Watrous2018Theory} or highly symmetric ensembles~\cite{Eldar2002Quantum,eldar2004optimal,leditzky2022optimality,Zhou2025Distinguishability}, no closed-form solution is known for the optimal POVM.
	Nevertheless, the PGM is near-optimal: $\eta_{\mathrm{opt}}^2 \le \eta_{\mathrm{PGM}} \le \eta_{\mathrm{opt}}$~\cite{barnum2002reversing,Watrous2018Theory}, where $\eta_{\mrm{opt}}$ and $\eta_{\mrm(PGM)}$ represents the success probabilities of the optimal POVM and the pretty good measurement respectively.  
	Now consider the variant based on (square-root) fidelity:
	\begin{equation}\label{Eq:FidelityStateDisc}
		\begin{aligned}
			\mathrm{maximize:} \quad & \sum_{i=1}^n \mathrm{F}(P_i,E_i) \\
			\mathrm{subject~to:} \quad & (E_1,\dots,E_n) \text{ is a POVM}.
		\end{aligned}
	\end{equation}
	This problem is exactly an ensemble projection problem, and is solved at the PGM.
	
	If all $P_i$ are rank-one, $P_i=\kb{p_i}$, then
	$\langle P_i,E_i\rangle = \la p_i,E_ip_i \ra =\mathrm{F}(P_i,E_i)^2$, and standard discrimination maximizes total (\textit{squared}) fidelity.
	Equivalently, the fidelity-based state discrimination problem~\eqref{Eq:FidelityStateDisc} maximizes the sum of square roots of probabilities, for which the PGM is optimal.
	This interpretation extends to arbitrary-rank states via Uhlmann’s theorem~\cite{uhlmann1976transition, Uhlmann2011Transition, Watrous2018Theory}, which allows Problem~\ref{Eq:FidelityStateDisc} to be written as an optimization over overlaps of purifications:
	\begin{equation}
		\begin{aligned}
			\text{maximize:} \quad & \sum_{i=1}^{n} |\langle u_i, v_i \rangle| \\
			\text{subject to:} \quad 
			& |u_i\rangle \in \mathrm{Pur}_{\mathcal Z}(P_i),
			|v_i\rangle \in \mathrm{Pur}_{\mathcal Z}(E_i), \text{ and } \sum_{i=1}^{n} E_i = I_{\mathcal X},
		\end{aligned}
	\end{equation}
	for some Hilbert space $\mathcal{Z}$ of appropriate dimension. 
	
	In summary, the PGM is the Bures projection of a weighted ensemble onto the POVM set, making it optimal for state discrimination based on square-root fidelity.
	The fact that square-root fidelity is equivalent to the square-root of probabilities explains why the PGM is generally close to optimal for state discrimination, yet truly optimal only in highly symmetric cases.
	
	\subsection{Geometry of the Petz Recovery Map} \label{Sec:PetzMapGeo}
	In this section, using fidelity/Bures projections (equivalently, prior-channel decompositions) we provide novel geometric interpretations to the Petz recovery map.
	The Petz recovery (transpose) map, originally introduced in a celebrated result by Petz studying sufficiency~\cite{petz1988sufficiency, petz1986sufficient} of quantum (Umegaki) relative entropy, is a central tool in quantum information.
	It is a quantum analogue of Bayes' theorem~\cite{Leifer2013towards} (also see~\cite[Eq. 12.34]{Wilde2013}), finds applications in quantum error correction~\cite{barnum2002reversing, ng2010simple,Mandayam2012Unified,Li2025Optimalitya}, as a decoder in quantum communication~\cite{beigi2016decoding, Burri2025Entanglement}, quantum sufficiency~\cite{petz1988sufficiency, fawzi2015quantum, wilde2015recoverability, sutter2016strengthened, jenvcova2017preservation, junge2018universal}, and black-hole physics~\cite{Cotler2019Entanglement, Kibe2022Holographic} among others.
	
	For a \textit{forward} channel $\Phi \in \cpt$ and a (forward) \textit{prior} $\rho \in \dx$, the associated Petz recovery map $\tilde \Phi_\rho \in \cpty$ is defined as
	\begin{equation}
		\tilde \Phi_\rho := \kp_\rho \circ \Phi^\dagger  \circ \kp_{\Phi(\rho)} \iv. 
	\end{equation}
	
	We now give a novel geometric interpretation for the Petz recovery map, in terms of Bures projections, which partly explains the (near-)optimality it achieves in a variety of settings.
	Moreover, it also generalizes the information-geometric derivation of the Bayes rule~(see the discussion in the Introduction of~\cite{bai2025quantum}) to the quantum setting.
	We will make use of results from Section~\ref{Sec:CPDecomp}, namely prior-channel decompositions, CP extensions of channels, and their relation to the fidelity/Bures projections.
	\begin{tbox}	
		\begin{theorem}
			Let $\Phi \in \cpt, \rho \in \dx,$ and $\Theta := \Phi \circ \kpr \in \cpxy$ be the CP extension of $\Phi$ by $\rho$.  
			The Petz map $\tilde \Phi_\rho$ associated with $(\Phi, \rho)$ is the Choi--Bures projection of $\Theta^\dagger$ to the set of (reverse) channels $\mrm{CPT}(\y, \x)$:
			\begin{equation} \label{Eq:108}
				\tilde \Phi_\rho = \Pi^{\mrm{CB}}_{\mrm{CPT}(\y, \x)}[\Theta^\dagger ] = \Pi^{\mrm{CB}}_{\mrm{CPT}(\y, \x)}[\kp_{\rho} \circ \Phi^\dagger ],
			\end{equation}
			where $\mrm{\Pi}^\mrm{CB}_{\cpty}$ is the Choi--Bures projection onto $\cpty$.
			Equivalently,
			\begin{equation}
				\fc(\tilde  \Phi_\rho) = \Pi_{\mrm{Tr}_\x, \iy} [\fc (\Theta)^\ic] = \Pi_{\mrm{Tr}_\x, \iy} [\fc (\Phi)^\ic \st X \rho].
			\end{equation}
		\end{theorem}
	\end{tbox}
	\begin{proof}
		By the equivalence between Choi--Bures projection and prior-channel decomposition (Proposition~\ref{Prop:ChPrDecomp}), it suffices to prove the first equality.
		Moreover, it also suffices to show that the channel component of the CP map $\Theta^\dagger  \equiv (\Phi \circ \kp_\rho)^\dagger  = \kp_\rho \circ \Phi^\dagger  \in \mrm{CP}(\y, \x)$ is the Petz map.
		To this end, we compute the prior $(\Theta^\dagger )^\dagger (I_\x) = \Phi(\rho)$ and compute the channel component as
		\begin{equation}
			\Theta^\dagger  \circ \kp_{\Phi(\rho)} \iv  = \kp_\rho \circ \Phi^\dagger  \circ \kp_{\Phi(\rho)} \iv  = \tilde \Phi_\rho,
		\end{equation}
		which is exactly the Petz map associated with the pair $(\Phi, \sigma)$. 
		This concludes the proof.
	\end{proof}
	The above result shows that the Petz map is the nearest reverse channel to the adjoint of the CP extension of $\Phi$ by $\rho$. 
	Here, proximity is measured in terms of fidelity / Bures distance between Choi matrices, or equivalently, minimal Frobenius distance between Stinespring operators.   
	
	\subsection{Information Geometry of the Leifer--Spekkens Quantum State Over Time}
	\label{Sec:LSFormalism}
	In this section, we show that prior-channel decompositions, and hence Bures projections, provide an information-geometric underpinning for a particular \emph{quantum state over time} (QSOT) formalism, namely that of Leifer and Spekkens (LS-QSOT)~\cite{leifer2006quantum, Leifer2007Conditional, Leifer2013towards}. 
	Our presentation is intentionally informal; for a precise and comprehensive treatment of the LS-QSOT formalism, we refer to~\cite{Leifer2013towards}.
	
	A QSOT is a quantum analogue of a classical joint probability distribution, whose aim is to extend the rules of classical probability theory to the quantum setting. 
	While probability vectors and conditional distributions generalize naturally to density matrices and quantum channels, respectively, combining these objects into a \emph{quantum joint state} is non-unique and leads to distinct QSOT constructions.
	
	Several QSOT formalisms have been proposed~\cite{Leifer2013towards, Fullwood2022Quantum, Fitzsimons2015Quantum}. 
	Relations between different QSOTs and associated quantum Bayes’ rules are discussed in~\cite{Parzygnat2023TimeReversal}, which implies that rotated Petz recovery maps~\cite{wilde2015recoverability} and their convex combinations~\cite{sutter2016strengthened} also induce QSOTs. 
	For broader perspectives on QSOTs, see~\cite{Parzygnat2023TimeReversal, Lie2024Quantum, Fullwood2025Dynamical}.

	\paragraph{Preliminaries.}
	LS-QSOT makes extensive use of the \textit{Jamiołkowski isomorphism}, which, like the Choi isomorphism, is a bijection between linear maps and bipartite matrices.  
	The Jamiołkowski isomorphism $\mathfrak{J} : \mathrm{LM} (\mathcal H, \mathcal K) \to \mathbb M_{\mathcal H \otimes \mathcal K} $ and its inverse are defined as
	\begin{equation}
		\mathfrak{J}(\Phi) := (\mathrm{Id}_\mathcal H \otimes \Phi)(\Omega) = \sum_{i,j=1}^{\mathrm{dim}(\mathcal H)} \ketbra{i}{j} \otimes \Phi(\ketbra{j}{i}) \quad \text{and} \quad \mathfrak{J}^{-1}(M)(H) = [M (H \otimes I_\mathcal K)]_\mathcal K.
	\end{equation}
	Though lacking some of the appealing properties of the Choi isomorphism (mainly, it does not map CP maps to the PSD cone), the Jamiołkowski isomorphism is basis-independent. 
	
	We now recall some elementary results from \textit{classical} probability theory.
	Let $\msf{A,B}$ denote \textit{elementary regions} with associated (finite) sample spaces $\mathcal{A,B}$ respectively.
	The set of all probability vectors associated with $\msf A$ and joints probability vectors over $\msf A \times \msf B$ are defined as:
	\begin{equation}
		\dsta := \left\{p_{\msf{A}} \in \mathbb{R}^{|\cl{A|}}_+ : \sum_{a \in \cl{A}} p_{\msf{A}}(a) = 1\right\} \qaq     \dstab := \left\{p_{\msf{AB}} \in \mathbb{R}^{\cl{|A||B|}}_+ : \sum_{(a,b) \in \cl{A \times B}} p_{\msf{AB}}(a,b) = 1\right\}.
	\end{equation}
	The standard bijection between $\dstab$ and $\dst_{\msf{BA}}$ is obtained by swapping the order of the elementary regions. 
	That is, $p_{\msf{AB}} \in \dstab \iff p_{\msf{BA}} \in \dst_{\msf{BA}}$ where $\pab (a,b) = p_{\msf{BA}}(b,a)$.   
	We make the bijection implicit and treat $\pab$ and $p_{\msf{BA}}$ as the same. 
	From a joint-distribution, $\pab \in \dstab$, one may construct the marginals $\pa \in \dsta$ and $p_{\msf B} \in \dstb$ as follows:
	\begin{equation}
		\pa(a) = \sum_{b \in \mathcal{B}} \pab (a,b) \qaq p_{\msf{B}}(b) = \sum_{a \in \mathcal{A}} \pab (a,b).
	\end{equation}
	A \textit{conditional distribution} over $\msf{B}$ conditioned on $\msf{A}$ is an $\mathcal A$-indexed collection of distributions over $\msf B$:
	\begin{equation}
		\mrm{CDist}_{\msf {B|A}} := \left\{\{ p^a_{\msf{B|A}}\}_{a \in \mathcal{A}} \subset \dstb \text{ for all }a \in \mathcal{ A}  \right\}.
	\end{equation}
	For each conditional distribution $\{ p^a_{\msf{B|A}}\}_{a \in \mathcal{A}} \in \mrm{CDist}_{\msf{B|A}}$, one may \textit{stack} the vectors to obtain $p_{\msf{B|A}}  := \sum_{a \in \mathcal{A}} e_a \otimes  p_{\msf{B|A}}^a,  $
	where $\{e_a\}_{a \in \mathcal{A}}$ is the standard basis over $\mathbb R^{|\mathcal{A}|}$. 
	Given such a (\textit{stacked}) conditional $p_{\msf{B|A}}$ and a prior distribution $\pa \in \dsta$, one can construct a joint distribution 
	\begin{equation} \label{Eq:ClassicalJoint}
		\pab := p_{\msf{B|A}} \odot (\pa \otimes 1_\msf{B}),
	\end{equation}
	where $1_\msf B := \sum_{b \in \mathcal{B}} e_b$ is the \textit{all-ones} vector of length $|\cl{B}|$ and $\odot$ denotes the Schur product. 
	The above operation can be written element-wise as $\pab (a,b) = p^a_{\msf{B|A}}(b) \pa (a)$ for all $(a,b) \in \mathcal{A \times B}$. 
	From a joint-distribution $\pab$, one can construct the conditional $p_{\msf{B|A}}$ as $\pab \odot (\pa \iv \otimes 1_\msf{B}) = p_{\msf{B|A}}$, provided every component of $\pa$ is strictly positive. 
	Otherwise, one restricts the support of the distributions.
	
	We will condense~\eqref{Eq:ClassicalJoint} as $\pab = p_{\msf{B|A}} \pa$, which, along with the bijection between $\pab$ and $p_{\msf {BA}}$, allow us to write
	\begin{equation}
		p_{\msf{B|A}}p_{\msf{A}} = p_{\msf{AB}} = p_{\msf{BA}} = p_{\msf{A|B}}p_{\msf{B}}.
	\end{equation}
	From this equation one can construct the marginals $p_\mathsf{A}, p_\mathsf{B}$, the conditionals $p_\mathsf{A|B}, p_\mathsf{B|A}$, and the \textit{Bayes' rule} $p_{\msf{B|A}} = p_{\msf{A|B}} p_{\mathsf B}/p_\mathsf{A}$.
	Essentially, it is this expression QSOTs aim to \textit{quantize}.
	
	\paragraph{Leifer--Spekkens QSOT.}
	
	Informally, a QSOT is a mapping of the form $\mathrm{QSOT} : \cpt \times \dx \to \mathbb{H}_{\x \ot \y}$.
	That is, given an \textit{input state} and a CPT map, construct a \textit{joint} Hermitian representation. 
	
	\begin{dbox}
		\begin{definition}[Leifer--Spekkens QSOT]
			The Leifer--Spekkens QSOT~\cite{leifer2006quantum, Leifer2007Conditional, Leifer2013towards} is defined as
			\begin{equation}\label{Eq:LSQSOT}
				\mathrm{QSOT}_{\mathrm{LS}}(\Phi,\rho)
				\equiv \Phi \ols \rho
				:= (\rho \hf \ot \iy)\,\fj(\Phi)\,(\rho \hf \ot \iy)
				= \fj(\Phi)\st{X}\rho,
			\end{equation}
			where $\fj(\Xi)$ denotes the Jamiołkowski representation of a linear map $\Xi$.
			See also~\cite[Def.~3.11, arXiv version]{Fullwood2025Dynamical}.
		\end{definition}
	\end{dbox}
	The (Hermitian) matrix $\Phi \ols \rho$ may fail to be PSD as it involves the Jamiołkowski representation of the CPT map $\Phi$. 
	Moreover,~\eqref{Eq:LSQSOT} resembles the CP extension $\fc (\Phi) \st{X} \rho^\ic$ of $\Phi$ by $\rho$ (see~\eqref{Eq:CPExt}), except for the transpose and the fact that it is the Choi representation that is used there. 
	This is not a coincidence as we now show that they are equivalent as the two are (Jamiołkowski/Choi) representations of the CP map $\Phi \circ \kpr$, which, as per our previous discussions, is the CP extension of $\Phi$ by $\rho$.
	\begin{tbox}
		\begin{prop}
			Let $\Phi \in \cpt, \rho \in \dx$ and $\Theta := \Phi \circ \kp_{\rho}$ be the CP extension. 
			It holds that
			\begin{equation}
				\Phi \ols \rho =  \fj(\Phi) \st{X} \rho = \fj(\Phi \circ \kp_{\rho}) = \mathfrak{J}(\Theta).
			\end{equation}
			Consequently, up to (the Jamiołkowski-)isomorphism, the LS-QSOT is equivalent to CP-extension.  
		\end{prop}
	\end{tbox}
	\begin{proof}
		The action of a linear map $\Xi$ can be recovered from its Jamiołkowski matrix as follows: $\Xi(X) = [\fj(\Xi) \dt{X} X]_{\mathcal Y}$. 
		Thus, it suffices to show that $[(\fj(\Phi) \st{X} \rho) \dt{X} X]_\mathcal Y = \Phi\left (\rho \hf X \rho \hf\right ) = \Theta(X),$
		for any $X \in \mx$.
		Indeed, this follows as
		\begin{equation}
			\begin{aligned}
				{[(\fj(\Phi) \st{X} \rho) \dt{X} X]}_\mathcal Y 
				&= [\fj(\Phi) \dt{X} (\rho \hf X \rho \hf )]_\mathcal Y = \Phi(\rho \hf X \rho \hf) = \Theta (X),
			\end{aligned}
		\end{equation}
		where, in the first equality, we leveraged the cyclicity of the partial trace.
		The bijection of the Jamiołkowski isomorphism allows us to conclude the proof. 
	\end{proof} Due to the relation $\Phi \ols \rho = \fj(\Phi \circ \kpr) \equiv \fj(\Theta)$, we will also refer to $\Theta$ as the \textit{joint state} (of $\Phi$ and $\rho$). 
	Essentially, we work at the Choi representation of the CP maps, whereas the LS formalism mostly utilizes the Jamiołkowski representation. 
	We now see how the LS-QSOT (\eqref{Eq:LSQSOT}) `quantizes' the basic operations of classical probability.

	\begin{tbox}
		\begin{theorem}[LS-QSOT as a quantum Bayes rule (\cite{Leifer2013towards, Parzygnat2023TimeReversal})]
			Let $\rho \in \dx$ be a prior state and a forward quantum channel, playing the roles of
			$p_\msf{A}$ and $p_{\msf{B|A}}$, respectively. 
			Let $\Theta := \Phi \circ \kp_R$
			\begin{equation}
				\Phi \ols \rho := \fj(\Phi)\st{X}\rho
				= \fj(\Phi \circ \kpr)
				\equiv \fj(\Theta),
			\end{equation}
			be the LS-QSOT. 
			The following statements hold:
			\begin{enumerate}
				\item (\emph{Joint state.}) 
				$\Phi \ols \rho$ is the quantum analogue of the classical joint distribution
				$p_{\msf{AB}} = p_{\msf{B|A}}p_\msf{A}$.
				\item (\emph{Marginals.}) 
				The prior and posterior are recovered by marginalization:
				\begin{equation}
					[\Phi \ols \rho]_\x = \rho,
					\qquad
					[\Phi \ols \rho]_\y = \Phi(\rho).
				\end{equation}
				\item (\emph{Conditionals.})
				The forward and reverse conditionals are obtained by star-division:
				\begin{equation}
					\fj(\Phi) = (\Phi \ols \rho)\st{X}\rho\iv,
					\qquad
					\fj(\Psi) = (\Phi \ols \rho)\st{Y}\Phi(\rho)\iv,
				\end{equation}
				where $\Psi \in \mrm{CPT}(\y,\x)$ is the reverse channel.
				
				\item (\emph{Quantum Bayes' rule.})
				The reverse channel $\Psi$ coincides with the Petz recovery map $\tilde \Phi_\rho = \mrm{K}_\rho \circ \Phi^\dagger \circ \mrm{K}_{\Phi(\rho)}^{-1},$ and satisfies the quantum analogue of Bayes' rule
				\begin{equation}
					\fj(\tilde \Phi_\rho) =  \fj(\Phi)\st{X}\rho  \st{Y} \Phi(\rho) \iv
					= \fj(\Theta) \st{Y} \Phi(\rho) \iv.        
				\end{equation}
			\end{enumerate}
		\end{theorem}
	\end{tbox}
	\begin{proof}
		See \cite{Leifer2013towards, Parzygnat2023TimeReversal}.
	\end{proof}

	\paragraph{Information geometry of the LS-QSOT.}
	We now arrive at our contributions, namely, we show that central operations of the LS-QSOT formalism, described in the above theorem, all have natural and compelling Bures-geometric interpretations. 
	We will be working with the Choi representations of the (CP and CPT) maps involved.
	\begin{tbox}
		\begin{theorem}
			Let $\Phi \in \cpt$, $\rho \in \dx$, and $\Theta := \Phi \circ \kp_{\rho} \in \cpxy$.
			Observe that $\Phi \ols \rho = \fj (\Theta) = \fj (\Phi) \st{X} \rho$. 
			The operations from LS-QSOT have the following geometric interpretations.
			\begin{enumerate}
				\item (\emph{Joint state.}) Construction of the joint state is equivalent to the Choi--Bures projection of $\Phi$ to the set of CP maps with prior $\rho$:
				\begin{equation}
					\fj \iv(\Phi \ols \rho) = \Theta = \Pi_{\mrm{CP}_\rho}^{\mrm{CB}}[\Phi].			
				\end{equation} 
				\item (\emph{Conditionals.})  The forward and reverse channels (conditionals) $\Phi \in \cpt$ and $\Psi \in \cpty$ coincide with the Choi--Bures projections of $\Theta$ and $\Theta^\dagger$ to the corresponding sets:
				\begin{equation}
					\Phi = \Pi^{\mrm{CB}}_{\cpt}[\Theta] \qaq \Psi = \Pi^{\mrm{CB}}_{\cpty}[\Theta^\dagger].
				\end{equation}
				\item (\emph{Quantum Bayes' rule.}): The reverse conditional $\Psi$ coincides with the Petz map $\tilde \Phi_\rho$: 
				\begin{equation}
					\tilde \Phi_\rho = \Psi = \Pi^{\mrm{CB}}_{\cpty} [\Theta ^\dagger]. 
				\end{equation}
				Thus, the LS-QSOT quantization of the expression $p_{\msf{A|B}}  = p_{\msf{B|A}}p_{\msf{A}}/p_\msf{B}$ has an information-geometric interpretation in terms of Choi--Bures projections. 
			\end{enumerate}
		\end{theorem}
	\end{tbox}
	\begin{proof}
		Each statement follows directly from the results relating prior-channel decompositions to Choi--Bures projections
		Section~\ref{Sec:CPDecomp} and the geometric characterization of the Petz map in
		Section~\ref{Sec:PetzMapGeo}.
	\end{proof}

	Our results provide a strong information-geometric underpinning to the Leifer--Spekkens formalism. 
	All the operations they define can thus be seen as (Bures-) geometrically \textit{optimal} quantizations of the basic rules of classical probability. 
	
	Recall that the Choi- and Jamiołkowski-isomorphisms have certain disadvantages (Basis-dependence and non-positivity for CP maps, respectively). 
	However, by taking a representation-independent perspective of these objects, one can circumvent these disadvantages, while still having access to the whole machinery (in bipartite settings).
    Whenever a concrete representation is required, one may equivalently employ the Choi, Jamiołkowski, or even Stinespring representation.

	\subsection{Geometry of the Quantum Minimal Change Principle} \label{Sec:QuantumMinimalChange}
	\textcite{bai2025quantum}, study the problem of \textit{quantum minimal change} with the motivation of generalizing Bayes' rule to the quantum setting using information-geometric principles. 
	We defer to~\cite{bai2025quantum} for further details.
    
	In this section, using prior-channel decompositions (equivalently, fidelity/Bures projections), we derive a shorter proof of their central result and provide the geometric interpretations for the same. 
	The problem statement and the central result are encapsulated in the following theorem.
	\begin{tbox}
		\begin{theorem}[Quantum minimal change principle~\cite{bai2025quantum}] \label{Thm:QMinPrinciple}
			Let $\Phi \in \cpt$, $\rho \in \dx$, and $\sigma \in \mathbb D_\y$. 
			The optimization problem
			\begin{equation}
				\begin{aligned}
					\text{maximize}: & \quad \mathrm{F}(\fc(\Phi) \st{X} \rho^{\intercal}, (\fc(\Psi) \st{Y} \sigma^{\intercal})^{\intercal}) \\
					\text{subject to}: & \quad \Psi \in \mathrm{CPT}(\cal Y, \cal X),  
				\end{aligned}
			\end{equation}
			is solved at 
			\begin{equation} \label{Eq:122}
				\Psi = \kpr \circ \Phi^\dagger \circ \kp_{(\Phi (\rho) \iv \# \sigma)^2} \circ \kp_\sigma \iv.
			\end{equation} 
			Moreover, if $[\Phi(\rho),\sigma] = 0$, then $\Psi = \tilde \Phi_\rho$ is the Petz map associated with the pair $(\Phi, \rho)$. 
		\end{theorem}
	\end{tbox}
	Using Definition~\ref{Def:CPExt}, we rewrite the objective function as $\mrm{F}(\fc(\Phi \circ \kp_{\rho}), \fc((\Psi \circ \kp_{\sigma})^\dagger ))$. 
	The extra transpose on the second argument is to compare the adjoint of the reverse process to the forward process.
	
	Denote $\Theta \equiv \Phi \circ \kpr$. 
	By the uniqueness of prior-channel decompositions (Proposition~\ref{Prop:ChPrDecomp}) and the invariance of fidelity under transpose of its arguments, the above optimization problem can be rewritten as 
	\begin{equation}
		\begin{aligned}
			\text{maximize}: & \quad \mathrm{F}(\fc(\Theta ^\dagger), \fc(\Xi)) \\
			\text{subject to}: & \quad \Xi \in \mathrm{CP}_\sigma(\cal Y, \cal X),
		\end{aligned}
	\end{equation}
	where $\mrm{CP}_\sigma (\y, \x) := \{\Xi \in \mrm{CP}(\y,\x): \Xi^*(I_\x) = \sigma\}$ is the subset $\mrm{CP}(\y, \x)$ with prior $\sigma$ (see~\eqref{Eq:CPPriorSet}).
	Observe that $\tr{\fc(\Xi)} = \tr{\sigma} = 1$ for all $\Xi \in \mrm{CP}_{\sigma}(\y, \x)$, which implies that the above optimization is equivalent to the Choi--Bures projection $\Pi_{\mrm{CP}_\sigma(\y, \x)}^{\mrm{CB}}[\Theta^\dagger]$. 
	
	This identification and the equivalence between Choi--Bures projections and prior-channel decompositions allow us to employ the results of this article to derive a simpler proof for Theorem~\ref{Thm:QMinPrinciple}.  
	We will only use the prior-channel decompositions, thereby demonstrating their versatility. 
	In~\cite{bai2025quantum}, the closed-form solution for the optimal channel $\Psi$ was derived using the technique of Lagrange multipliers.
	Our results provide a simpler proof and geometric interpretations for this result. 
	
	\begin{tbox}
		\begin{theorem}
			Let $\Theta := \Phi \circ \kpr \in \cpxy$ be the CP-extension of $\Phi$ by $\rho$. 
			The optimal reverse channel $\Psi \in \cpty$ (\eqref{Eq:122}) is the \emph{channel component} of the Choi--Bures projection of $ \Theta^\dagger$ to $\mrm{CP}_\sigma(\y, \x)$:
			\begin{equation} 
				\Psi =  \left(\Pi_{\cpty}^{\mrm {CB}} \circ \Pi^{\mrm{CB}}_{\mrm{CP}_\sigma(\y, \x)}\right)[\Theta^{\dagger} ]
			\end{equation}
		\end{theorem}
	\end{tbox}
	\begin{proof}
		Using Proposition~\ref{Prop:CBProjPriorS}, we first project $\Theta^\dagger$ to $\mrm{CP}_{\sigma}(\y, \x)$:
		\begin{equation}
			\Xi \equiv \Pi^{\mrm{CB}}_{\mrm{CP}_\sigma(\y, \x)}[\Theta^{\dagger} ] = \Theta^\dagger \circ \mrm{K}_{(\Phi(\rho) \iv \# \sigma)^2}, 
		\end{equation}
		where we used the fact that the prior of $\Theta^\dagger$ is $\Theta(\ix) = \Phi(\rho)$. 
		We project $\Xi$ to $\cpty$, which is simply the channel component of $\Xi$:
		\begin{equation}
			\Pi_{\cpty}^{\mrm{CB}}[\Xi] = \Xi \circ \kp_{\Xi^\dagger(I_\x)} \iv = \Theta^\dagger \circ \mrm{K}_{(\Phi(\rho) \iv \# \sigma)^2} \circ \kp_{\sigma}\iv = \Psi
		\end{equation}
		where, in the second equality, we used the fact that $\Xi ^\dagger(I_\x) = \sigma$. 
		The last equality follows from the relation $\Theta^\dagger = \kpr \circ \Phi^\dagger$ and comparison with \eqref{Eq:122}.
		This concludes the proof.
	\end{proof}
	
	Thus, the results of this article greatly simplify the proof for the central result of~\cite{bai2025quantum}. 
	Moreover, it also elucidates the geometric aspect, which is understood as follows.
	We first project (with respect to the Choi--Bures distance) $\Theta ^\dagger$ to $\mrm{CP}_\sigma (\y, \x)$, the set of CP maps with prior $\sigma$.
	Let $\Xi$ be this projection. 
	Then the optimal reverse channel $\Psi$ is simply the channel component of $\Xi$, which can be obtained via the channel-prior decomposition.
	At the level of Choi representations, one can see $\Psi^\dagger $ as the result of three projections starting from $\Phi$.
	\begin{equation}
		\fc(\Phi) \mapsto \Pi_{\mrm{Tr}_\y, \rho ^\ic}[\fc(\Phi)] = \fc (\Theta) \mapsto \Pi_{\mrm{Tr}_\x, \sigma} [\fc(\Theta)] \equiv \fc(\Xi^\dagger ) \mapsto \Pi_{\mrm{Tr}_\x, \iy} [\fc(\Xi^\dagger )] = \fc(\Psi^\dagger ).
	\end{equation}

	We conclude with the following remark.
	Theorem~\ref{Thm:QMinPrinciple} involves combining a channel ($\Phi$) and a prior $(\rho)$ into a joint state $\fc(\Phi) \st{X} \rho^\ic$, and thus can be seen as a quantum state over time. 
	Moreover, this state over time representation can be seen as the Choi matrix of the CP extension of the channel-state pair involved: $	 \fc(\Phi) \st{X} \rho^\intercal = \fc(\Phi \circ \kpr)$.
	This further means that the QSOT of~\textcite{bai2025quantum} is, by Choi isomorphism, equivalent to the CP extension of the channel by the prior, and also to the Leifer--Spekkens QSOT~(see Section~\ref{Sec:LSFormalism}).

	\subsection{Geometry of Channel Distance Measures}
	The final application we present is a geometric interpretation of various distance measures defined over quantum channels. 
	For two quantum channels $\Phi, \Psi \in \mathrm{CPT} (\mathcal X, \mathcal Y)$, they can be defined as follows. 
	\begin{equation}
		\tilde{\mathrm{D}}(\Phi, \Psi) := \sup_{\rho \in \mathbb{D}_\mathcal X} \tilde {\mathrm{D}}_\rho(\Phi, \Psi) := \sup_{\rho \in \mathbb{D}_\mathcal X} \mathrm{D}((\mathrm{Id}_\mathcal X \otimes \Phi)(\Omega_\rho), (\mathrm{Id}_\mathcal X \otimes \Psi)(\Omega_\rho)),
	\end{equation}
	where $\mathrm{D}$ is a distance measure (a bona fide distance or a divergence) defined over positive semidefinite matrices and $\Omega_\rho  = \ketbra{\omega_\rho}{\omega_\rho}$ is the canonical purification of $\rho \in \dx$.  
	Some instances are diamond norm distance (when $\mathrm{D}$ is chosen to be the trace distance)~\cite{kitaev1997quantum, Watrous2018Theory}, entanglement-fidelity distance (when $\mathrm{D}$ is chosen to be Bures distance)~\cite{gilchrist2005distance, beny2010general, Schumacher1996sending}, and generalized channel divergences (when $\mathrm{D}$ is chosen to be a divergence satisfying DPI)~\cite{leditzky2018approaches}. 
	
	It is easy to see from the definition that all of these channel distances can be seen as maximizing the corresponding matrix distances between the Choi matrices of the CP-extensions of the channels over all possible reference priors:
	\begin{equation}
		\mathrm{\tilde D}_\rho(\Phi, \Psi) = \mathrm{D}(\fc(\Phi \circ \kpr), \fc(\Psi \circ \kpr)).
	\end{equation}
	
	Indeed, the Bures projection aspect also provides a geometric meaning to these distances: 
	\begin{equation}
		\mathrm{\tilde D}_\rho(\Phi, \Psi) = \mathrm{D}(\fc(\Phi) \st{X} \rho^\intercal , \fc(\Psi) \st{X} \rho^\intercal) = \mathrm{D}(\Pi_{\mathcal X, \rho^\intercal}[\fc(\Phi)], \Pi_{\mathcal X, \rho^\intercal}[\fc(\Psi)]). 
	\end{equation}
	Geometrically,  $\mrm{D}_\rho$ is the distance (divergence) between the Bures projections of $\fc(\Phi)$ and $\fc(\Psi)$ onto the spectrahedron with \textit{input marginal} $\rho^\ic$. 
	Thus, the distance measure $\mathrm{\tilde D}$ corresponds to maximally separating the Choi--Bures projections of $\Phi$ and $\Psi$ over all spectrahedra as defined by the priors.

	\section{Conclusion} \label{Sec:Conclusion}
	In this article, we considered the \textit{fidelity/Bures projection} problem, where we project, with respect to fidelity (equivalently, Bures distance), PSD matrices to various spectrahedra of interest in quantum information.
	These spectrahedra are defined as the PSD preimage of a given PSD matrix under a given quantum channel. 
	
	We demonstrated that the Bures projection is unique on the support of the input matrix, and for certain channels, this projection can be computed in a simple closed-form manner, including the important case of the partial trace. 
	We related the saturation of DPI for fidelity to the existence of the closed-form and provided simple sufficient conditions to check whether a projection problem can be solved by our closed-form, and then derived explicit solutions for a variety of channels, including partial trace, pinching maps, orthogonal-projective measurements, and the projection of PSD ensembles.
	
	We discussed the notion of \textit{prior–channel decompositions}, an intuitive and geometrically optimal way to (uniquely) decompose any CP map into a prior and a channel component. 
	This decomposition manifests in several of the applications discussed thereafter.
	
	Various applications of these projections are then discussed, with a particular emphasis on those related to partial trace. 
	These include novel geometric interpretations for the Petz recovery map, Pretty good measurement, Leifer-Spekkens QSOT, various distance measures between channels, and the \textit{minimal change} approach for quantizing the Bayes' theorem. 
	Due to the optimality of projections, our results provide yet another justification for useful tools like PGM, Petz map, and the \textit{partial normalization} operation. 
	
	Our results open up various open problems, some of which we list below. 
	\begin{itemize}
		\item Identifying additional classes of quantum channels for which the $\Gamma$ map provides a closed-form expression for the fidelity projection.
		\item Exploring concrete applications of projection techniques for channels beyond the partial trace, such as the pinching map.
		\item Applying the closed-form Bures projection to quantum process tomography, and extending the Bayesian state tomography framework of~\cite{Afham2022Quantum} to the process tomography setting.
	\end{itemize}
	
	\section{Acknowledgments}
	AA thanks Chris Ferrie, Roberto Rubboli, Francesco Buscemi, Valerio Scarani, and Ion Nechita for helpful discussions.
	AA also thanks Valerio Scarani and Ion Nechita for hosting him during research visits in which parts of this work were carried out.
	A part of this work is present in AA's Ph.D. thesis~\cite{Afham2025Thesis}, and he thanks the reviewers for helpful comments.
	AA and MT are supported by the NRF Investigatorship award (NRF-NRFI10-2024-0006).
		
	\printbibliography	

	\appendix
	\section*{Appendix}
\addcontentsline{toc}{section}{Appendix}

\renewcommand{\thesection}{\Alph{section}}

\section*{A$\quad$Choi matrix of the adjoint map}
\label{App:ChoiMatrixAdjoint}
\addcontentsline{toc}{subsection}{A. Choi matrix of the adjoint map}
We are interested in proving the following statement.
\begin{tbox}
\begin{prop}
    Let $\Phi \in \mathrm{HP}(\x, \y)$ be a Hermitian preserving map. 
    It holds that $\fc(\Phi^{\dagger}) = \fc(\Phi)^\ic$, where the transpose is taken in the basis used to define the Choi matrix.
\end{prop}
\end{tbox}
\begin{proof}
	Let $\fc(\Phi)^\intercal \in \mathbb H_{\mathcal X \otimes \mathcal Y}$ define the linear map $\Psi$ from $\mathbb M_\mathcal Y$ to $\mathbb M_\mathcal X$ with the action:
	\begin{equation}
		\Psi(Y) := [\fc(\Phi)^\intercal (I_\mathcal X  \otimes Y^\intercal )]_\mathcal X, 
	\end{equation}
	for any $Y \in \mathbb M_\mathcal Y$.
	We aim to show that $\Psi = \Phi^{\dagger}$, which is equivalent to showing that $\langle Y, \Phi(X) \rangle  = \langle \Psi(Y), X \rangle, $
	for all $X \in \mathbb M_\mathcal X$ and $Y \in \mathbb M_\mathcal Y$.
	Observe that
	
	\begin{equation}
		\begin{aligned}
			\left \langle X, \Psi(Y) \right\rangle &= \operatorname{Tr}\left[X^{\dagger} \Psi(Y)\right] \\ 
			&=  \operatorname{Tr}\left[[\fc(\Phi)^\intercal (I_\mathcal X  \otimes Y^\intercal )]_\mathcal X X^{\dagger}\right] & \text{definition of $\Psi(Y)$}\\
			&= \operatorname{Tr}\left[\fc(\Phi)^\intercal (I_\mathcal X \otimes Y^\intercal ) (X^{\dagger} \otimes I_\mathcal Y ) \right] & \text{adjoint of $\mathrm{Tr}_\mathcal Y$}\\
			&= \operatorname{Tr}\left[\fc(\Phi)^\intercal (X^{\dagger} \otimes Y^\intercal )\right] \\
			&= \operatorname{Tr}\left[\fc(\Phi) (\bar{X} \otimes Y)\right]  & \operatorname{Tr}[A^\intercal] = \operatorname{Tr}[A] \text{ for any $A$}\\
			&= \operatorname{Tr}\left[\fc(\Phi) (\bar{X} \otimes I_\mathcal Y) (I_\mathcal X  \otimes Y) \right] \\
			&= \operatorname{Tr}\left[[\fc(\Phi) (\bar{X} \otimes I_\mathcal Y)]_\mathcal Y Y \right] & \text{adjoint of tensoring $I_\mathcal X$} \\
			&= \operatorname{Tr}\left[[\fc(\Phi) ({(X^{\dagger})}^\intercal  \otimes I_\mathcal Y)]_\mathcal Y Y \right] \\
			&= \operatorname{Tr}[\Phi(X^{\dagger}) Y]\\
			&= \operatorname{Tr}[\Phi(X)^{\dagger} Y] & \text{$\Phi$ is Hermitian-preserving} \\
			&= \langle \Phi(X), Y \rangle,
		\end{aligned}
	\end{equation}
	for any $X \in \mathbb M_\mathcal X$ and $Y \in \mathbb M_\mathcal Y$, which implies $\Psi = \Phi^{\dagger}$.
    Here we used $\bar X$ to denote the complex conjugate of $X$.
	Thus, for any HP map $ \Phi$, we have $\fc(\Phi)^\intercal =\fc(\Phi^{\dagger})$ as claimed.
\end{proof}

\section*{B$\quad$Proof: restricted uniqueness of Bures projection} \label{App:ProjectionUniquenessCombined}
\addcontentsline{toc}{subsection}{B. Proof of the restricted uniqueness of fidelity projections}
\begin{tbox}
\begin{theorem}[Theorem \ref{Thm:ProjectionUniquenessCombined} restated]\label{AppThm:ProjectionUniquenessCombined}
		Let $(\Lambda, P, C)$ constitute a projection problem and let $Q_0, Q_1 \in \plc$ be arbitrary projections. 
		It holds that
        \begin{equation}
            P^0 Q_0 P^0 = P^0 Q_1 P^0,
        \end{equation}
        where $H^0$ denotes the orthogonal projector onto the support of $H$ for any Hermitian $H$. 
        Moreover, if $P$ is full-rank, then the projection is unique. 
	\end{theorem}
\end{tbox}
\begin{proof}
    We first prove it for the case where $P$ is full-rank.
    We split the (sub)proof into two cases: the case when there is at least a single full-rank projection, and the case where every projection is rank-deficient. 
    
    Consider the first case where there exists a full-rank projection $Q_0 \in \Pi_{\Lambda, C}[P]$. 
    To see that this implies the uniqueness of the projection, assume there exists another (possibly rank-deficient) projection $Q_1 \in \Pi_{\Lambda, C}[P]$. 
    Let $Q_t \equiv (1-t)Q_0 + t Q_1$ for $t \in [0,1]$ and observe that $Q_t \in \Pi_{\Lambda, C}[P]$ for all $t \in [0,1]$, as square-root fidelity is concave and the feasible set is convex, and thus any convex combination of projections is necessarily a projection.
    For any $t \in [0,1)$, it holds that $Q_t$ is full-rank, and thus there exist multiple (an infinite number of) full-rank projections. 
    This is, however, not possible due to the strict concavity\footnote{A function $f:\mathcal{S} \to \mathbb R$ over a convex set $\mathcal{S}$ is \textit{concave} if $f((1-t)x + ty) \geq (1-t)f(x) + t f(y)$ for all $x,y \in \mathcal{S}$ and $t \in [0,1]$. We say $f$ is \textit{strictly concave} if $f$ is concave and the above inequality is strict for all $t \in (0,1)$ and distinct $x,y \in \mathcal{S}$~\cite{Boyd2004}.} of fidelity over full-rank states~\cite{bhatia2018strong}. 
    Hence, any full-rank projection must necessarily be unique.  
    
    Now consider the case where all projections are rank-deficient. 
    If all the projections have the same support, then we can restrict ourselves to this support and again use the fact that fidelity is strictly concave over full-rank states.
    
    Thus, let us consider the case where there exist projections $Q_0, Q_1 \in \Pi_{\Lambda, C}[P]$ with distinct supports. 
    The support of a non-trivial convex combination $Q_t$ contains the supports of $Q_0$ and $Q_1$. 
    Restricting to this support would make $Q_t$ (for $t \in (0,1)$) effectively full-rank, and thus necessarily unique, 
    which implies that if $P$ is full-rank, then the projection problem has a unique solution.
    
    Now we consider the case where $P$ is rank-deficient. 
    Recall that fidelity can be restricted to the support in the following manner~\cite[Proposition 3.12]{Watrous2018Theory}:
    \begin{equation}
        \mrm{F}(A,B) = \mrm{F}(A,A^0B A^0) = \mrm{F}(B^0 A B^0, B).  
    \end{equation}  
    Let $Q_0, Q_1 \in \plc$ be distinct projections. 
    It holds that
    \begin{equation}
        \mrm{F}(P, P^0Q_0P^0) = \mrm{F}(P, Q_0) = \mrm{F}(P, Q_1) = \mrm{F}(P, P^0Q_1P^0), 
    \end{equation}
    where the second equality follows from the fact that $Q_0$ and $Q_1$ are both projections.
    Restricting ourselves to $\mrm{supp}(P)$, we have $P$ to be effectively full-rank and thus the Bures projection (in $\mrm{supp}(P)$) must be unique by the previous part of this proof, whence the claim follows.   
\end{proof}

\end{document}